\documentclass[preprint,3p,12pt,number,sort&compress]{elsarticle}
\input{Qcircuit}
\usepackage{appendix,amsmath,amsthm,amssymb,bbm,graphicx,hyperref,nccmath,subfigure}

\newtheorem{theorem}{Theorem}
\newtheorem{corollary}{Corollary}[theorem]

\begin{document}

\title{Efficient quantum circuits for Szegedy quantum walks}

\author{T. Loke and J.B. Wang\corref{cor1}}
\cortext[cor1]{jingbo.wang@uwa.edu.au}
\address{School of Physics, The University of Western Australia, Perth WA 6009,
Australia}

\begin{abstract}
	
A major advantage in using Szegedy's formalism over discrete-time and continuous-time quantum walks lies in its ability to define a unitary quantum walk on directed and weighted graphs. In this paper, we present a general scheme to construct efficient quantum circuits for Szegedy quantum walks that correspond to classical Markov chains possessing transformational symmetry in the columns of the transition matrix. In particular, the transformational symmetry criteria do not necessarily depend on the sparsity of the transition matrix, so this scheme can be applied to non-sparse Markov chains. Two classes of Markov chains that are amenable to this construction are cyclic permutations and complete bipartite graphs, for which we provide explicit efficient quantum circuit implementations. We also prove that our scheme can be applied to Markov chains formed by a tensor product. We also briefly discuss the implementation of Markov chains based on weighted interdependent networks. In addition, we apply this scheme to construct efficient quantum circuits simulating the Szegedy walks used in the quantum Pagerank algorithm for some classes of non-trivial graphs, providing a necessary tool for experimental demonstration of the quantum Pagerank algorithm.

\end{abstract}

\maketitle

\section{Introduction}
\label{sec:introduction}

In the last two decades, quantum walks have produced a wide variety of quantum algorithms to address different problems, such as the graph isomorphism problem \cite{douglas_classical_2008,gamble_two-particle_2010,berry_two-particle_2011}, ranking nodes in a network \cite{paparo_google_2012,paparo_quantum_2013,loke_comparing_2015}, and quantum simulation \cite{lloyd_universal_1996,berry_black-box_2012,berry_exponential_2014,berry_hamiltonian_2015,berry_simulating_2015}. There are two main types of quantum walks, the discrete-time quantum walk \cite{aharonov_quantum_2001, tregenna_controlling_2003} and the continuous-time quantum walk \cite{farhi_quantum_1998, gerhardt_continuous-time_2003}. In order to achieve the promised processing power of quantum walk based algorithms, it is necessary to formulate quantum circuits that can be efficiently realized on a quantum computer to simulate the corresponding quantum walks. For the discrete-time quantum walk, quantum circuit implementations have been realized on some classes of highly symmetric graphs \cite{douglas_efficient_2009, loke_efficient_2012} and sparse graphs \cite{jordan_efficient_2009}. Analogous work on quantum circuit implementations has also been carried out for continuous-time quantum walks, particularly for sparse graphs \cite{childs_exponential_2003, berry_efficient_2007, childs_limitations_2010}, and special classes of graphs \cite{qiang_efficient_2016}.

A different formalism for the discrete-time quantum walk was introduced by M. Szegedy (termed as the Szegedy quantum walk) by means of quantizing Markov chains \cite{szegedy_spectra_2004, szegedy_quantum_2004}. A major advantage in using Szegedy's formalism is that it allows one to define a unitary quantum walk on directed (or more generally, weighted) graphs \cite{paparo_quantum_2013}, which the standard discrete-time and continuous-time quantum walk formalisms do not allow for directly. It has also been shown that the Szegedy quantum walk provides a quadratic speedup in the quantum hitting time (for any reversible Markov chain with one marked element), compared to the classical random walk \cite{krovi_quantum_2015}. The Szegedy quantum walk has also proven to be useful for a variety of different applications, for example, verifying matrix products \cite{buhrman_quantum_2006}, testing group commutativity \cite{magniez_quantum_2007}, searching for a marked element in a graph \cite{magniez_search_2011}, approximating the effective resistance in electrical networks \cite{wang_quantum_2013}, and quite notably, it is instrumental in the quantum Pagerank algorithm \cite{paparo_google_2012, paparo_quantum_2013} for determining the relative importance of nodes in a graph.

As with the discrete-time and continuous-time quantum walk, the realization of quantum algorithms based on Szegedy walks requires an efficient quantum circuit implementation for the walk itself. Previous work by Chiang \emph{et al.} \cite{chiang_efficient_2010} has showed that given a transition matrix that corresponds to an arbitrary sparse classical random walk, a quantum circuit for the walk operator can be realized. These quantum circuits scale linearly with the sparsity parameter $d$ and polynomially in $\mbox{log}(1/\epsilon)$ (where $\epsilon$ is the approximation accuracy to the entries of the transition matrix). In this paper, we consider some families of transition matrices that possess certain symmetries where an efficient quantum circuit for the Szegedy evolution operator can also be realized, even if they are not sparse.

The paper is organised as follows. In section \ref{subsec:background}, we present the mathematical definition of a Szegedy quantum walk. In section \ref{subsec:circuit}, we establish a theoretical framework for the construction of efficient quantum circuit implementations of Szegedy quantum walks, given a Markov chain that has a transition matrix satisfying some constraints. We then apply this theory to the class of Markov chains where the transition matrix is described by a cyclic permutation in section \ref{sec:cyclic}. We also show that the class of complete bipartite graphs can be efficiently simulated in section \ref{sec:kmn}. Then, we extend the results of section \ref{subsec:circuit} to a tensor product of Markov chains in section \ref{sec:tensor}. We briefly discuss the implementation of weighted interdependent networks in section \ref{sec:win}. Next, in section \ref{sec:google}, we apply the results of section \ref{sec:cyclic} to construct efficient quantum circuits simulating Szegedy walks involved in the quantum Pagerank algorithm for some classes of directed graphs. Lastly, we draw our conclusions in section \ref{sec:conclusion}.

\section{Theory}
\label{sec:theory}

\subsection{Background}
\label{subsec:background}

A classical Markov chain is comprised of a sequence of random variables $ X_t $ for $ t \in \mathbb{Z_+}$ with the property that $P(X_t|X_{t-1},X_{t-2},...,X_1)=P(X_t|X_{t-1})$, i.e. the probability of each random variable is only dependent on the previous one. Suppose that every random variable $X_t$ has $N$ possible states, i.e. $ X_t \in \{s_1,\ldots,s_N\} $. If the Markov chain is time-independent, that is: 

\begin{equation}
P(X_t|X_{t-1})=P(X_{t'}|X_{t'-1}) \mbox{ } \forall t,t' \in \mathbb{Z_+},
\end{equation} 

\noindent then the process can be described by a single (left) stochastic matrix $P$ (called the transition matrix) of dimension $N$-by-$N$, where $P_{i,j} = P(X_t=s_i|X_{t-1}=s_j) $ is the transition probability of $s_j \rightarrow s_i$ with the column-normalization constraint $\displaystyle\sum_{i=1}^{N} P_{i,j} = 1$.

A random walk on a graph can be described by a Markov chain, with states of the Markov chain corresponding to vertices on the graph, and edges of the graph determining the transition matrix. For a graph $G(V,E)$ with vertex set $V=\left\{v_1,\ldots,v_N\right\}$ and edge set $E=\left\{(v_i,v_j),(v_k,v_l),\ldots\right\}$, the $N$-by-$N$ adjacency matrix $A$ is defined as:

\begin{equation}
A_{i,j} = 
\begin{cases}
1 & (v_i,v_j) \in E \\
0 & \mbox{otherwise}.
\end{cases}
\label{eqn:adj}
\end{equation}

From this, we can define the corresponding transition matrix for $G$ as:

\begin{equation}
P_{i,j}=A_{i,j}/\mbox{indeg}(j),
\label{eqn:pgraph}
\end{equation} 

\noindent where $\mbox{indeg}(j)=\displaystyle\sum_{i=1}^n A_{i,j}$ is the in-degree of vertex $j$. An undirected graph is a graph with a symmetric adjacency matrix ($A=A^T$), i.e. if $(v_i,v_j) \in E$ then necessarily $(v_j,v_i) \in E$ also. Directed graphs, on the other hand, do not have a symmetric adjacency matrix ($A\neq A^T$). Weighted graphs are a generalization of equation (\ref{eqn:adj}), allowing for any real value in any entry of the adjacency matrix (equivalent to a transition matrix).

Szegedy's method for quantizing a single Markov chain with transition matrix $P$ starts by considering a Hilbert space $\mathcal{H}=\mathcal{H}_1 \otimes \mathcal{H}_2$ composed of two registers of dimension $N$ each \cite{szegedy_spectra_2004, szegedy_quantum_2004}. A state vector $\ket{\Psi} \in \mathcal{H}$ of dimension $N^2$ can thus be written in the form $ \ket{\Psi} = \displaystyle\sum_{i=0}^{N-1} \sum_{j=0}^{N-1} a_{i,j} \ket{i,j} $. The projector states of the Markov chain are defined as:

\begin{equation}
\ket{\psi_i} = \ket{i} \otimes \displaystyle\sum_{j=0}^{N-1} \sqrt{P_{j+1,i+1}} \ket{j} \equiv \ket{i} \otimes \ket{\phi_i},
\label{eqn:psi}
\end{equation}

\noindent for $i \in \{0,\ldots,N-1\}$. We can interpret $\ket{\phi_i}$ to be the square-root of the $i$th column of the transition matrix $P$. Note that $ \bra{\psi_{i'}} \psi_i \rangle = \delta_{i,i'} $ due to the orthonormality of basis states and the column-normalization constraint on $P$. The projection operator onto the space spanned by $\{\ket{\psi_i}:i \in \{0,\ldots,N-1\}\}$ is then:

\begin{equation}
\Pi = \displaystyle\sum_{i=0}^{N-1} \ket{\psi_i}\bra{\psi_i},
\label{eqn:proj}
\end{equation}

\noindent and the associated reflection operator $\mathcal{R} = 2\Pi - I$. Define also the swap operator that interchanges the two registers as:

\begin{equation}
\mathcal{S} = \displaystyle\sum_{i=0}^{N-1}\sum_{j=0}^{N-1}\ket{i,j}\bra{j,i},
\end{equation}

\noindent which has the property $\mathcal{S}^2 = I$. Then the walk operator for a single step of the Szegedy walk is given by:

\begin{equation}
U_{walk} = \mathcal{S}(2\Pi - I) = \mathcal{S} \mathcal{R}.
\label{eqn:Uwalk}
\end{equation}

We note here that these definitions are essentially a special case of the bipartite walks framework used by Szegedy originally \cite{szegedy_spectra_2004, szegedy_quantum_2004}. The above can be put into the same framework by considering the two-step walk operator $ U_{walk}^2 = (2S \Pi S - I )(2\Pi - I) $, which is equal to Szegedy's bipartite walk operator using equivalent reflections in both spaces. We use the above definitions since, for an undirected graph, one step of $U_{walk}$ is equivalent to one step of a discrete-time quantum walk using the Grover diffusion operator at every vertex as the coin operator \cite{childs_quantum_2008}. However, unlike the conventional discrete-time quantum walk, the Szegedy walk as defined above provides a unitary time-evolution operator even for directed and weighted graphs.

\subsection{Circuit implementation}
\label{subsec:circuit}

For the purposes of this paper, we adopt the usual definition of an efficient quantum circuit implementation. Namely, if the quantum circuit uses at most $O(\mbox{poly}(\mbox{log}(N)))$ one- and two- qubit gates to implement the Szegedy walk operator $U_{walk}$, it is considered as efficient \cite{douglas_efficient_2009,chiang_efficient_2010}.

\begin{figure}[htp]
	\centering
	\includegraphics[scale=0.20]{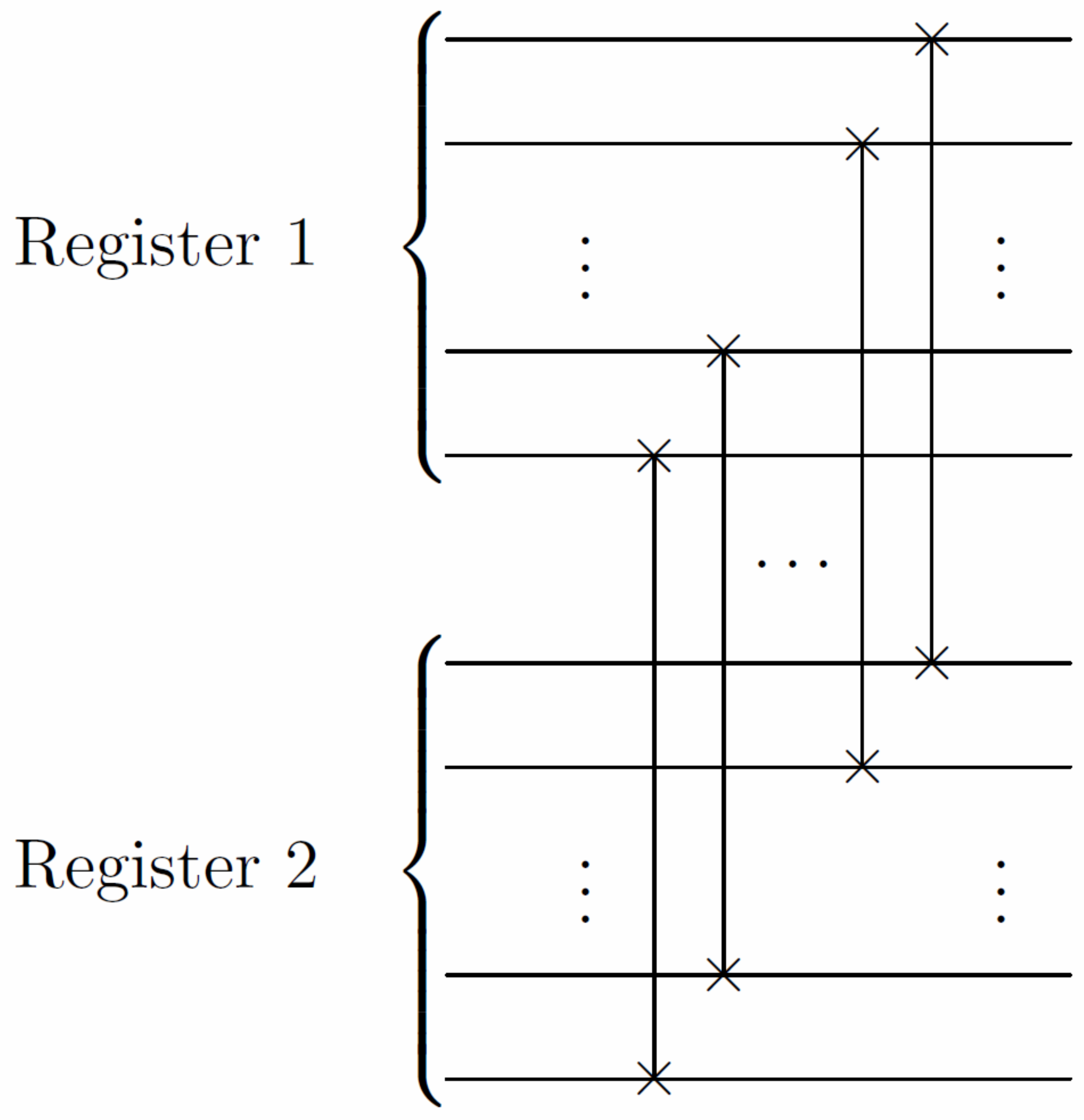}
	\caption{Quantum circuit implementing the swap operator $\mathcal{S}$.}
	\label{fig:Scirc}
\end{figure}

We first note that implementing the swap operator $\mathcal{S}$ in a quantum circuit is straightforward and efficient---it simply consists of $\lceil \mbox{log}_2(N) \rceil$ swap gates applied between the two registers, as shown in Figure \ref{fig:Scirc}. In order to implement the reflection operator $\mathcal{R}$, we diagonalize it using a unitary operation $U$:

\begin{equation}
U \mathcal{R} U^\dag = U (2\Pi - I) U^\dag = 2U\Pi U^\dag - I.
\end{equation}

Let $U$ be of the block-diagonal form $U = \displaystyle\sum_{i=0}^{N-1} \ket{i}\bra{i} \otimes U_i$. Then expanding the above expression using equations (\ref{eqn:psi}) and (\ref{eqn:proj}) gives:

\begin{equation}
U \mathcal{R} U^\dag = 2 \displaystyle\sum_{i=0}^{N-1} \ket{i}\bra{i} \otimes \left( U_i \ket{\phi_i} \bra{\phi_i} U_i^\dag \right) - I.
\end{equation}

Choose $U_i$ such that $U_i \ket{\phi_i} = \ket{b} \mbox{ } \forall i $ for some $\ket{b}$ which is a computational basis state, i.e. $U_i$ acts to transform $\ket{\phi_i}$ into a fixed computational basis state $\ket{b}$, or conversely, the inverse operator $U_i^\dagger$ generates the state $\ket{\phi_i}$ from $\ket{b}$. Then:

\begin{equation}
U \mathcal{R} U^\dag = 2 I_N \otimes \ket{b}\bra{b} - I = I_N \otimes (2\ket{b}\bra{b} - I_N) \equiv D,
\label{eqn:diag}
\end{equation}

\noindent which is a diagonal matrix that can be readily implemented using a controlled-$\pi$ gate (with multiple conditionals) applied to the second register. Hence, the Szegedy walk operator becomes:

\begin{equation}
U_{walk} = \mathcal{S} U^\dag D U.
\label{eqn:Uwalk2}
\end{equation}

It is, of course, infeasible to realize such an operator $U$ efficiently in general---however certain symmetries in the transition matrix $P$ do allow for $U$ to be implemented efficiently. Now, let us write $U_i = K_b^\dag T_i$, such that $K_b^\dag \ket{\phi_r} = \ket{b}$ and $T_i \ket{\phi_i} = \ket{\phi_r}$, where $K_b$ and $T_i$ are both unitary operations, and $ \ket{\phi_r} $ is a chosen reference state. In other words, this means that $K_b$ prepares the state $\ket{\phi_r}$ from a computational basis state $\ket{b}$, and $T_i$ transforms $\ket{\phi_i}$ into $\ket{\phi_r}$---this corresponds to transforming the $i$th column of $P$ into the $r$th column of $P$ (with square roots on both). Note that this scheme works for any choice of the reference state $\ket{\phi_r}$ or the computational basis state $\ket{b}$, i.e. they can be chosen arbitrarily.

%Let us write the initial state of the system in the form

%\begin{equation}
%\ket{\Psi} = \ket{\Psi^\perp} + \displaystyle\sum_{i=1}^N \langle \psi_i \ket{\Psi} \ket{\psi_i}
%\end{equation}

%\noindent where $\langle \psi_i \ket{\Psi^\perp} = 0 \mbox{ } \forall i $. Then from the definition of $R$, we can write

%\begin{equation}
%R\ket{\Psi} = -\ket{\Psi^\perp} + \displaystyle\sum_{i=1}^N \langle \psi_i \ket{\Psi} \ket{\psi_i}
%\end{equation}

%Now, if we were to apply $U$ to $\ket{\Psi}$ and expand the result, we obtain

%\begin{eqnarray*}
%U\ket{\Psi} &=& \left( \displaystyle\sum_{i=1}^{N} \ket{i} \bra{i} \otimes (K_b^\dag T_i) \right) \left( \ket{\Psi^\perp} + %\displaystyle\sum_{i=1}^N \langle \psi_i \ket{\Psi} \ket{i} \otimes \ket{\phi_i} \right)
%\end{eqnarray*}

\begin{figure}[htp]
	\centering
	\includegraphics[scale=0.30]{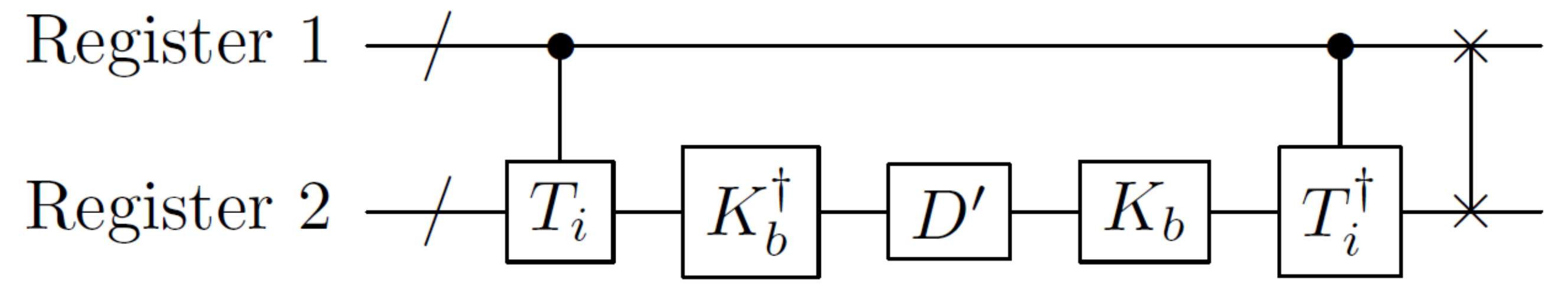}
	\caption{Quantum circuit implementing $U_{walk}$ according to equation (\ref{eqn:Uwalk2}), with $D' = 2 \ket{b}\bra{b} - I_N $. The symbol / denotes a shorthand for registers of qubits representing $N$ states each.}
	\label{fig:Uimpl}
\end{figure}

Figure \ref{fig:Uimpl} shows the general quantum circuit for implementing $U_{walk}$. The criteria for an efficient circuit implementation of $U_{walk}$ is thus:

\begin{theorem}
	The Szegedy walk operator $U_{walk}$ can be efficiently implemented if the following conditions are satisfied:
	\begin{enumerate}
		\item Able to implement all $N$ transformations in the form $\mathcal{T} = \left(\displaystyle\sum_{i=0}^{N-1} \ket{i}\bra{i} \otimes T_i\right)$ (and its inverse) where $T_i:\ket{\phi_i}\rightarrow\ket{\phi_r}$ using at most \textnormal{$O(\mbox{poly}(\mbox{log}(N)))$} gates; and
		\item Able to identify and implement a preparation routine $K_b:\ket{b}\rightarrow\ket{\phi_r}$ (and its inverse) for a chosen computational basis state $\ket{b}$ efficiently.
	\end{enumerate}
	\label{thm:mainres}
\end{theorem}

\begin{proof}
	From equation (\ref{eqn:Uwalk2}), $U_{walk}$ is efficiently implementable iff its components $\mathcal{S}$, $U^\dagger$, $D$ and $U$ are efficiently implementable. As discussed previously, $\mathcal{S}$ and $D$ are easily implementable. Since $U = \displaystyle\sum_{i=0}^{N-1} \ket{i}\bra{i} \otimes U_i$ and $U_i = K_b^\dag T_i$, we can write $U = (I_N \otimes K_b^\dagger) \left(\displaystyle\sum_{i=0}^{N-1} \ket{i}\bra{i} \otimes T_i\right)$. Hence $U$ and $U^\dagger$ (and thus $U_{walk}$) are efficiently implementable if the conditions above are satisfied.
\end{proof}

For condition 2 of Theorem \ref{thm:mainres} in particular, we can make use the state preparation method using integrals \cite{zalka_simulating_1998,mosca_quantum_2001,grover_creating_2002} in order to implement $K_b$ and its inverse. This gives the following result:

\begin{corollary}
	Suppose $\mathcal{T}$ (and its inverse) can be efficiently implemented, and that the probability distribution $\left\{P_{j,r}\right\}$ can be generated by integrating a certain probability density function $p(z)$, i.e. $P_{j,r} = \displaystyle\int_{z_j}^{z_{j+1}} p(z) dz$ where $z_{j+1}-z_j=\Delta z$. If $p(z)$ is (classically) efficiently integrable, then the Szegedy walk operator $U_{walk}$ can be efficiently implemented.
	\label{crl:mainres}
\end{corollary}

\begin{proof}
	This follows directly from \cite{grover_creating_2002}, which provides a method to prepare $\ket{\phi_r} = \displaystyle\sum_{i=0}^{N-1} \sqrt{P_{j,r}} \ket{j}$ under these conditions. The inverse procedure $K_b^\dagger$ can be obtained by uncomputation of $K_b$, i.e. applying the circuit for $K_b$ in reverse order. Hence, the conditions in Theorem \ref{thm:mainres} are met.
\end{proof}

However, we note that $K_b$ doesn't have to be constructed using this method---any $K_b$ that satisfies $K_b^\dag \ket{\phi_r} = \ket{b}$ is sufficient. Some examples of this will be provided in the following sections.

This formalism can be generalized slightly by allowing for more than one reference state. Consider a partition $ Z = \{ Z_1,\ldots,Z_M \} $ of the set $ X = \{ 0,\ldots,N-1 \} $ such that $ \emptyset \notin Z $, $ \displaystyle\bigcup_{x=1}^{M} Z_x = X $ and $ Z_x \cap Z_y = \emptyset \mbox{ } \forall x \neq y $. Generalizing the previous equations gives:

\begin{eqnarray}
\displaystyle
U_{walk} &=& \mathcal{S} \prod_{x=1}^{M} U_x^\dag D_x U_x \label{eqn:gen1} \\
U_x &=& \sum_{y\in Z_x} \ket{y}\bra{y} \otimes U_{x,y} + \left( \sum_{y' \notin Z_x} \ket{y'}\bra{y'} \right) \otimes I_N \label{eqn:gen2} \\
D_x &=& \left( 2 \sum_{y\in Z_x} \ket{y}\bra{y} \right) \otimes \ket{b_x}\bra{b_x} - I \label{eqn:gen3} \\
U_{x,y} &=& K_{b_x}^\dag T_{x,y}, \label{eqn:gen4}
\end{eqnarray}

\noindent with $ K_{b_x}^\dag \ket{\phi_{r_x}} = \ket{b_x} $ and $ T_{x,y} \ket{\phi_y} = \ket{\phi_{r_x}} $, where $ \ket{\phi_{r_x}} $ is the reference state for the subset $Z_x$. The corresponding quantum circuit is similar to that given in Figure \ref{fig:Uimpl}, except with $M$ independent segments with conditionals determined by the partition $Z$. The criteria for efficient circuit implementation of $U_{walk}$ with the partition $ Z = \{ Z_1,...,Z_M \} $ is thus:

\begin{theorem}
	The Szegedy walk operator $U_{walk}$ can be efficiently implemented with the partition $ Z = \{ Z_1,...,Z_M \} $ if the following conditions are satisfied:
	\begin{enumerate}
		\item The number of subsets must be bounded above as \textnormal{$M \leq O(\mbox{poly}(\mbox{log}(N)))$}.
		\item For every subset $Z_x$,
		\begin{enumerate}
			\item Able to implement controlled unitary operations according to $Z_x$ efficiently, i.e. $ \left(\sum_{y\in Z_x} \ket{y}\bra{y} \right) \otimes U + \left( \sum_{y' \notin Z_x} \ket{y'}\bra{y'} \right) \otimes I_N $ can be efficiently implemented assuming $U$ can be done efficiently;
			\item Able to implement all $|Z_x|$ transformations in the form $ \mathcal{T}_x = \sum_{y\in Z_x} \ket{y}\bra{y} \otimes T_{x,y} + \left( \sum_{y' \notin Z_x} \ket{y'}\bra{y'} \right) \otimes I_N $ where $ T_{x,y}: \ket{\phi_y} \rightarrow \ket{\phi_{r_x}} $ using at most \textnormal{$O(\mbox{poly}(\mbox{log}(N)))$} gates; and
			\item Able to identify and implement a preparation routine $K_{b_x}:\ket{b_x}\rightarrow\ket{\phi_{r_x}}$ (and its inverse) for a chosen computational basis state $\ket{b_x}$ efficiently.
		\end{enumerate}
	\end{enumerate}
	\label{thm:genres}
\end{theorem}

\begin{proof}
	From equation (\ref{eqn:gen1}), $U_{walk}$ is efficiently implementable iff each of its components $\mathcal{S}$, $U_x^\dagger$, $D_x$ and $U_x$ are efficiently implementable and the number of subsets is bounded above as \textnormal{$M \leq O(\mbox{poly}(\mbox{log}(N)))$}. As discussed previously, $\mathcal{S}$ is easily implementable. If condition 2(a) is satisfied, then it follows that $D_x$ can be efficiently implemented, since $U = 2\ket{b_x}\bra{b_x}-I_N$ can be implemented using a single controlled-$\pi$ gate. From equations (\ref{eqn:gen2}) and (\ref{eqn:gen4}), we can write:
	\begin{eqnarray*}
		\displaystyle
		U_x &=& \left( \left(\sum_{y\in Z_x} \ket{y}\bra{y} \right) \otimes K_{b_x}^\dagger + \left( \sum_{y' \notin Z_x} \ket{y'}\bra{y'} \right) \otimes I_N \right) \\
		&& \qquad \quad \left( \sum_{y\in Z_x} \ket{y}\bra{y} \otimes T_{x,y} + \left( \sum_{y' \notin Z_x} \ket{y'}\bra{y'} \right) \otimes I_N \right).
	\end{eqnarray*}
	If conditions 2(a) and 2(c) are satisfied, then the first term can be implemented by identifying $U = K_{b_x}^\dagger$ in condition 2(a). Hence $U_x$ and $U_x^\dagger$ (and thus $U_{walk}$) are efficiently implementable if the conditions above are satisfied.
\end{proof}

As before, we can use the state preparation method using integrals in order to implement $K_{b_x}$ and its inverse, although any other construction of $K_{b_x}$ that satisfies $K_{b_x}^\dag \ket{\phi_{r_x}} = \ket{b_x}$ is sufficient. 
%Using the state preparation method using integrals, a weaker form of Corollary \ref{crl:mainres} can be formulated:

%\begin{corollary}
%	Suppose $\mathcal{T}_x$ (and its inverse) can be efficiently implemented, and that the probability distribution $\left\{P_{j,r_x}\right\}$ can be generated by integrating a certain probability density function $p_x(z)$, i.e. $P_{j,r_x} = \displaystyle\int_{z_j}^{z_{j+1}} p_x(z) dz$ where $z_{j+1}-z_j=\Delta z$. If $p_x(z)$ is (classically) efficiently integrable, then the operator $U_x$ can be efficiently implemented.
%	\label{crl:genres}
%\end{corollary}

%\begin{proof}
%	Almost exactly the same as the proof for Corollary \ref{crl:mainres}, except with $U_x$ instead of $U_{walk}$.
%\end{proof}

%\section{Examples}
%\label{sec:examples}

%Here, we will examine some classes of transition matrices $P$ where the formalism discussed in section \ref{subsec:circuit} can be used to implement the Szegedy walk operator $U_{walk}$ efficiently.

%\subsection{Cyclic permutations}
%\label{subsec:cyclic}
\section{Cyclic permutations}
\label{sec:cyclic}

A cyclic permutation is where elements of a set are shifted by a fixed offset. Formally, let $R$ be the one-element right-rotation operator such that $ R \left( \{ c_1, c_2, ..., c_{N-1}, c_N \}^T \right) = \{ c_N, c_1, ..., c_{N-2}, c_{N-1} \}^T $, and $L = R^\dag$ is the one-element left-rotation operator. In this section, we consider classes of transition matrix $P$ where each column is a fixed cyclic permutation of the previous column, i.e. we can write $ \ket{\phi_{i+1}} = R^x \ket{\phi_i} $ or $ \ket{\phi_{i+1}} = L^x \ket{\phi_i} $ for some $ x \in \mathbb{Z}_+ $ for all values of $i$.

One important class of cyclic permutations is circulant matrices. In a circulant matrix, each column is right-rotated by one element with respect to the previous column, i.e. $ \ket{\phi_{i+1}} = R \ket{\phi_i} $. Hence, we can construct the projector states as:

\begin{equation}
\ket{\psi_i} = \ket{i} \otimes R^{i} \ket{\phi_0}.
\label{eqn:psicirc}
\end{equation}

Hence, if we select a single reference state as $\ket{\phi_r} = \ket{\phi_0}$, we can define the transformations as $T_i = (R^{\dag})^{i} = L^{i} $. This set of transformations can always be efficiently implemented, since the operations $\{L, L^2, L^4, L^8, ...\}$ can be realized by decrementing the corresponding bit values by 1---this takes $ O(\mbox{log}(N)^2) $ generalized CNOT gates. The inverse transformations $T_i^\dag = R^{i}$ can be implemented in a similar fashion by incrementing bit values by 1 where needed. The preparation routine $K_b$ for preparing $ \ket{\phi_0} $ from some computational basis state $\ket{b}$ can either be constructed by the state preparation method using integrals (see Corollary \ref{crl:mainres}) or by finding some other unitary transformation that does $\ket{b}\rightarrow\ket{\phi_r}$. Assuming one such $K_b$ can be efficiently implemented, then by Theorem \ref{thm:mainres} the Szegedy walk operator $U_{walk}$ for the circulant transition matrix can be efficiently implemented.

\begin{figure}[htp]
	\centering
	\includegraphics[scale=0.20]{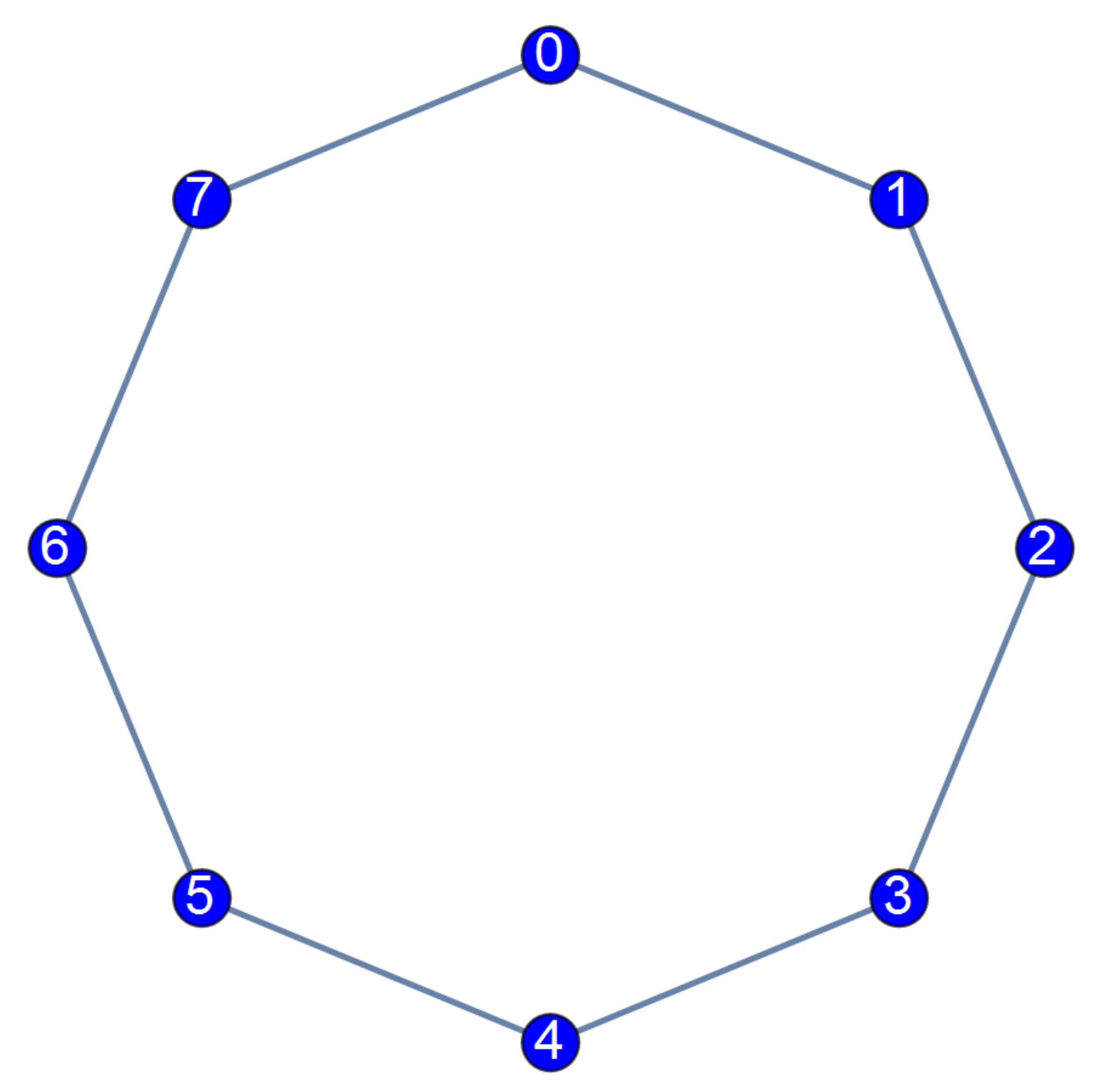}
	\caption{Cycle graph $C_N$ with parameter $N=8$.}
	\label{fig:C8}
\end{figure}

In the case of an undirected graph, that is, $A = A^T$, we can construct the Szegedy walk using the above method. For example, the cycle graph $C_N$ for some $N = 2^n$ (example shown in Figure \ref{fig:C8}) where $n \in \mathbb{Z}_+$ can be constructed by identifying $ \ket{\phi_0} = \{ 0, 1/\sqrt{2}, 0, ..., 0, 1/\sqrt{2} \}^T $, which can be prepared using a Hadamard gate and a sequence of controlled-NOT gates. Figure \ref{fig:CNimpl} shows the quantum circuit for $C_{N}$ with $\ket{b} = \ket{N/2-1}$.

\begin{figure}[htp]
	\centering
	\includegraphics[scale=0.3]{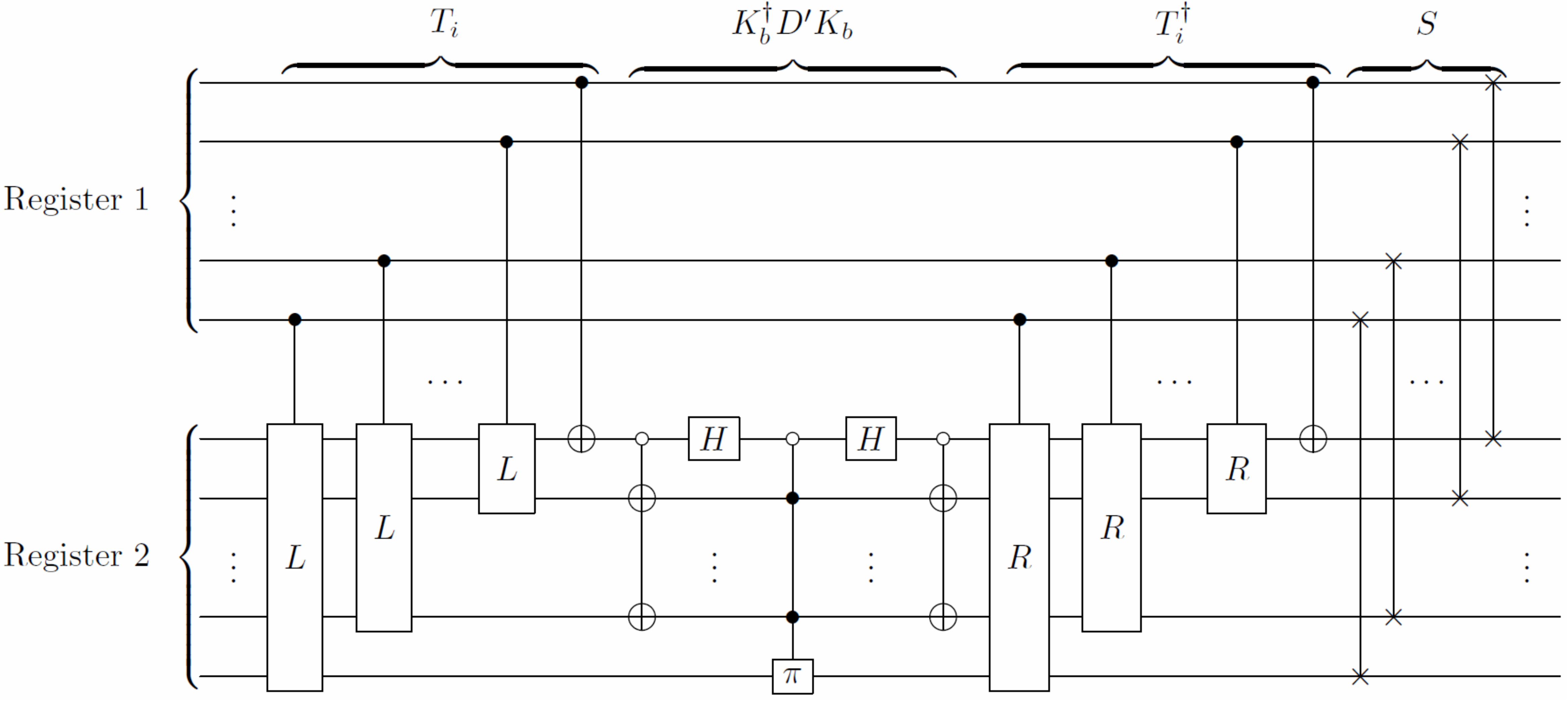}
	\caption{Quantum circuit implementing $U_{walk}$ for the $C_{N}$ graph.}
	\label{fig:CNimpl}
\end{figure}

\begin{figure}[htp]
	\centering
	\includegraphics[scale=0.20]{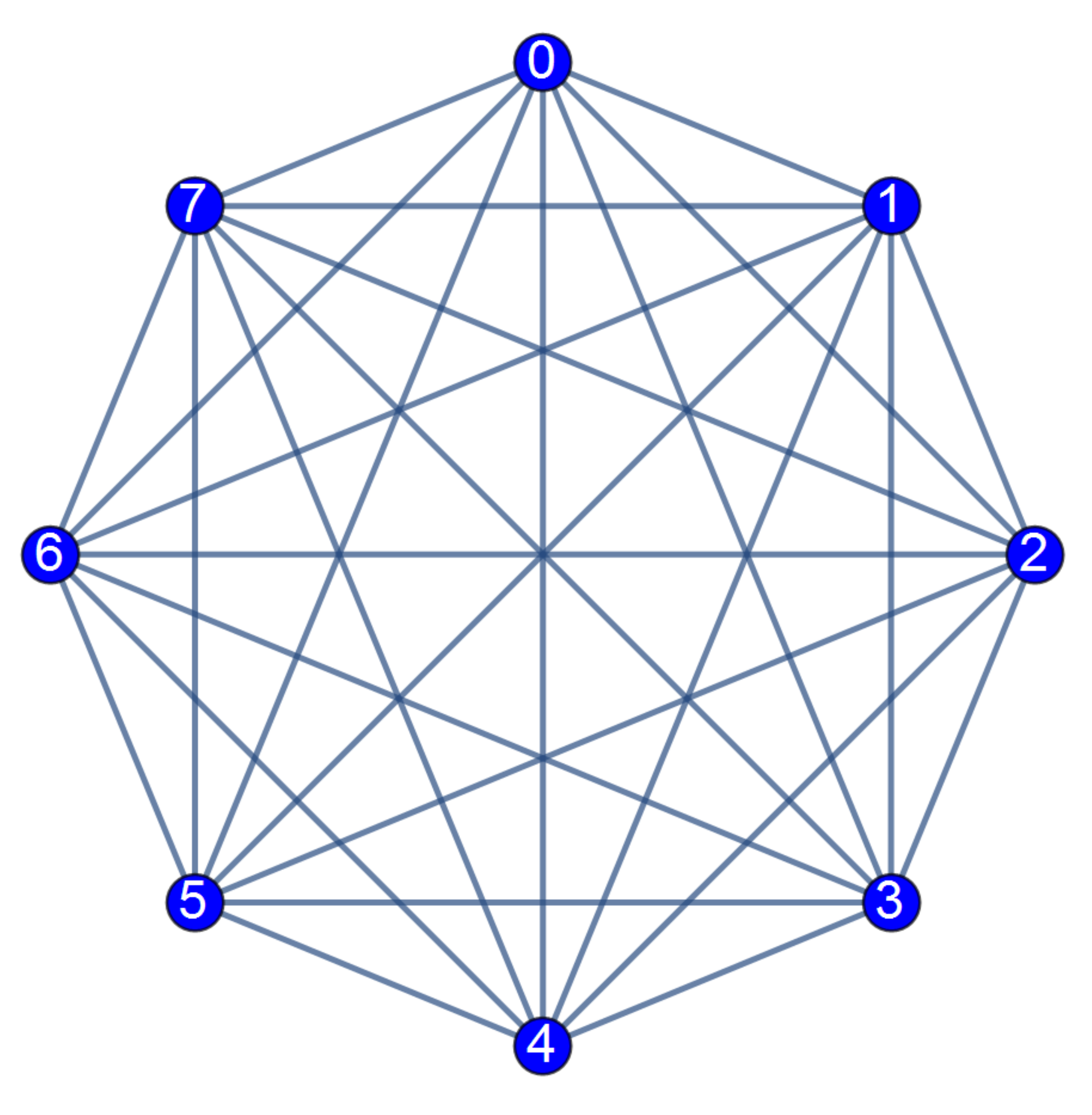}
	\caption{Complete graph $K_N$ with parameter $N=8$.}
	\label{fig:K8}
\end{figure}

More dense graphs, such as the complete graph $K_N$ for some $N = 2^n$ (example shown in Figure \ref{fig:K8})  where $n \in \mathbb{Z}_+$, which is circulant with $ \ket{\phi_0} = \{ 0, \allowbreak \sqrt{1/(N-1)}, \allowbreak \ldots, \allowbreak \sqrt{1/(N-1)} \}^T $, can be implemented in similar fashion using rotation gates for the preparation routine $K_b$, as shown in Figure \ref{fig:KnKb} for $\ket{b}=\ket{1}$, where the rotation gate $R_y$ is defined as:

\begin{equation}
R_y(\theta) = \left(
\begin{array}{c c}
\mbox{cos}(\theta) & -\mbox{sin}(\theta) \\
\mbox{sin}(\theta) & \mbox{cos}(\theta)
\end{array}
\right).
\end{equation}

Using this preparation routine, the quantum circuit implementing $U_{walk}$ for $K_N$ is shown in Figure \ref{fig:Knimpl}.

\begin{figure}[htp]
	\centering
	\includegraphics[scale=0.30]{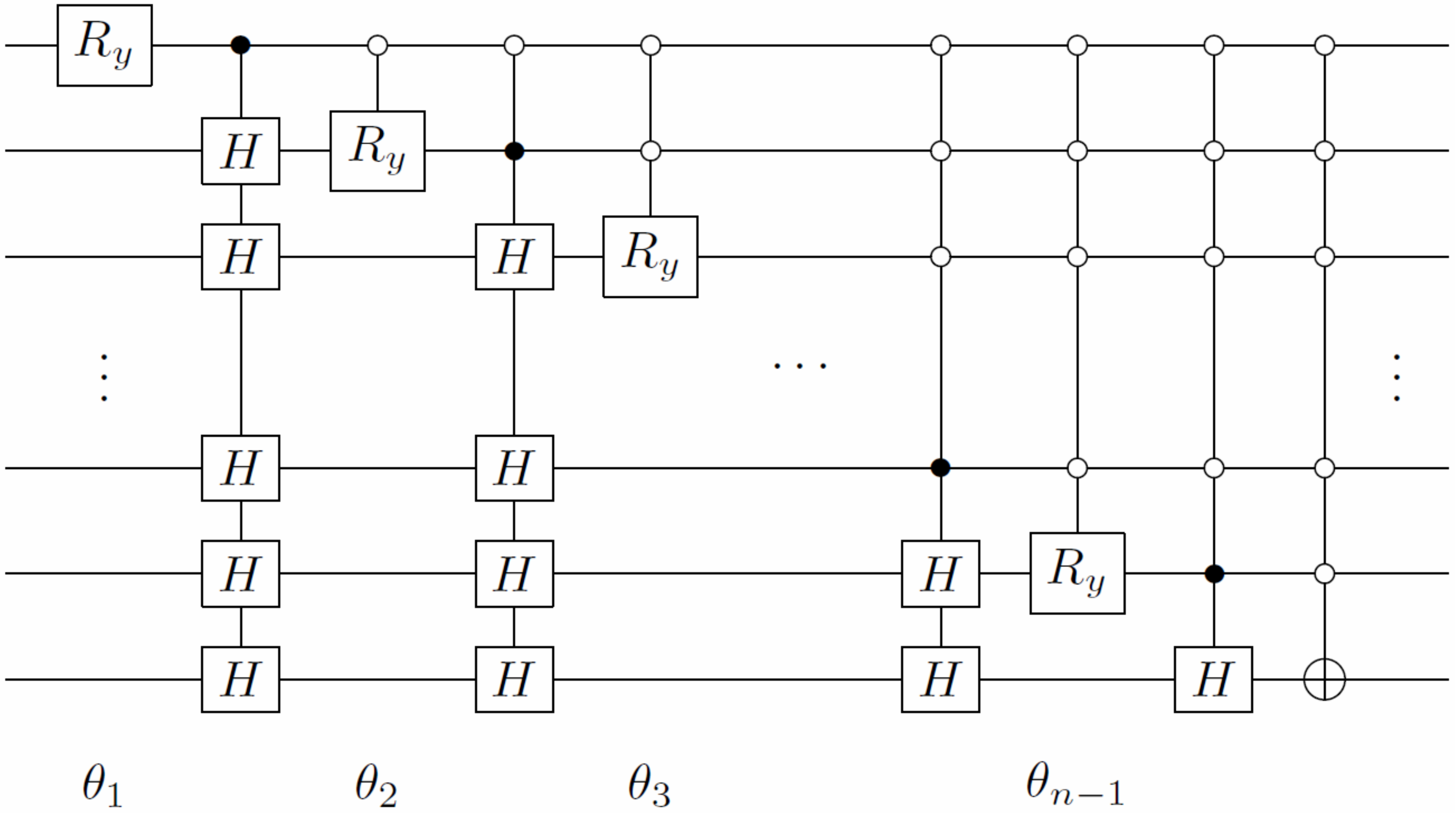}
	\caption{Quantum circuit implementing $K_{b}:\ket{b} \rightarrow \ket{\phi_0}$ for the complete graph $K_{2^n}$. The rotation angles are given by $\mbox{cos}(\theta_i) = \frac{2^{n-i}-1}{2^{n-i+1}-1}$.}
	\label{fig:KnKb}
\end{figure}

\begin{figure}[htp]
	\centering
	\includegraphics[scale=0.32]{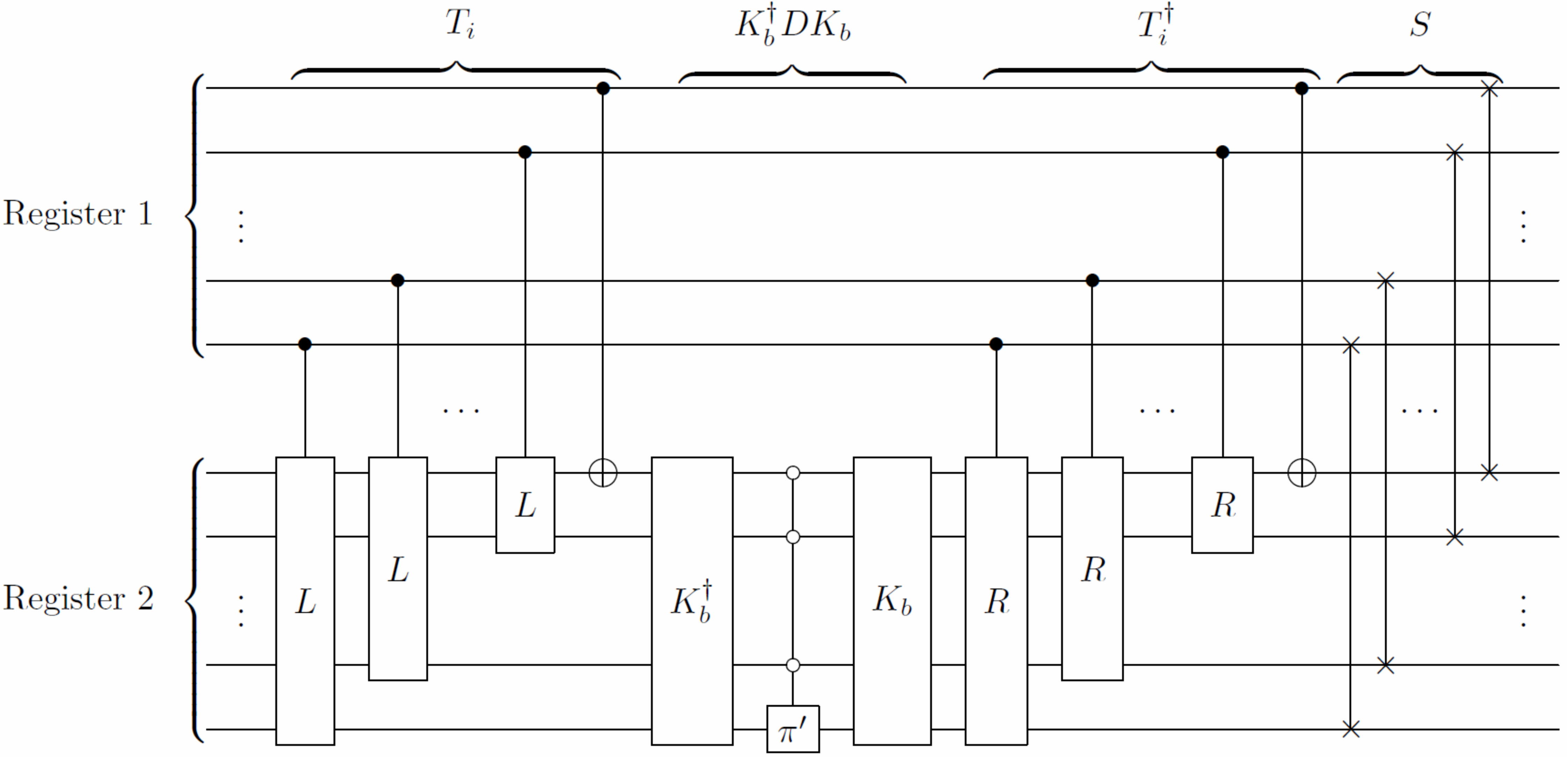}
	\caption{Quantum circuit implementing $U_{walk}$ for the $K_N$ graph, where the quantum circuit for the preparation routine $K_b$ is given in Figure \ref{fig:KnKb}. The $ \pi' $ gate is equivalent to $ \sigma_x \pi \sigma_x $, i.e. a $-1$ phase applied to the first state of the qubit.}
	\label{fig:Knimpl}
\end{figure}

Note, however, that we are not constrained to work with transition matrices $P$ which correspond to a Markov chain on an undirected graph---any kind of circulant matrix where $ \ket{\phi_r} $ can be efficiently prepared can be implemented in this fashion.

We can generalize this construction to any kind of $P$ where the columns are formed by any cyclic permutation, by changing equation (\ref{eqn:psicirc}) into:

\begin{equation}
\ket{\psi_i} = \ket{i} \otimes R^{ix} \ket{\phi_0},
\end{equation}

\noindent for some $x \in \mathbb{Z}_+$ (or similarly for left-rotations), and then changing the transformations $T_i$ accordingly to $T_i = (R^\dagger)^{ix} = L^{ix}$. Since $\{L, L^2, L^4, L^8, ...\}$ can be realized efficiently, it follows that any $L^x$ can be realized efficiently. Hence, every transformation of the form $T_i = (L^x)^{i} = L^{ix}$ can also be realized efficiently.

\section{Complete bipartite graphs}
\label{sec:kmn}

\begin{figure}[htp]
	\centering
	\includegraphics[scale=0.25]{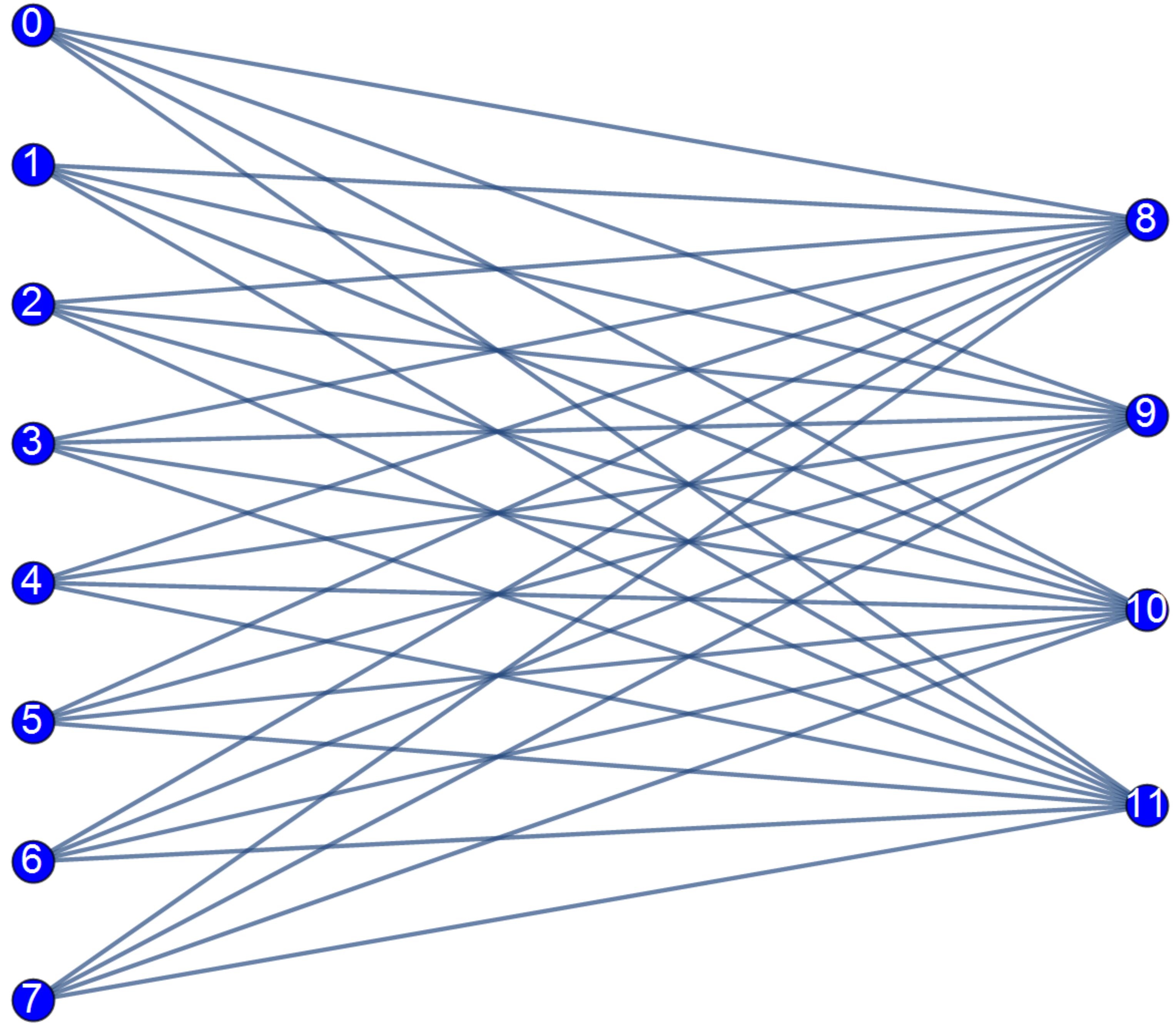}
	\caption{Complete bipartite graph $K_{N_1,N_2}$ with parameters $N_1=8$ and $N_2=4$.}
	\label{fig:kmn}
\end{figure}

One particular class of dense graphs that can be efficiently implemented using the construction in section \ref{subsec:circuit} is the class of complete bipartite graphs. A complete bipartite graph $K_{N_1,N_2}$ (example shown in Figure \ref{fig:kmn}) is a bipartite graph (i.e. the set of graph vertices is decomposable into two disjoint sets such that no two graph vertices within the same set are adjacent) such that every pair of graph vertices in the two sets are adjacent, with $N_1$ and $N_2$ vertices in the sets respectively. The transition matrix (i.e. the column-normalized adjacency matrix) of dimension $(N_1+N_2)$-by-$(N_1+N_2)$ for the graph $K_{N_1,N_2}$ is:

\begin{equation}
P(K_{N_1,N_2}) = \left(
\begin{array}{c c}
0 & \frac{1}{N_1}J_{N_1,N_2} \\
\frac{1}{N_2}J_{N_2,N_1} & 0
\end{array}
\right),
\label{eqn:pk}
\end{equation}

\noindent where $J_{a,b}$ denotes the all-1's matrix of dimension $a$-by-$b$. The graph suggests the natural partition $Z = Z_1 \cup Z_2$ where $Z_1 = \{0,\ldots,N_1-1\}$ and $Z_2 = \{N_1,\ldots,N_1+N_2-1\}$. The transformations $T_{x,y}$ for either set are all identical, i.e. $T_{x,y} = I_{N_1+N_2}$, since within each subset every column of the transition matrix is equal. From equation (\ref{eqn:pk}), we can identify $ \ket{\phi_0} = \{0,\ldots,0,1/\sqrt{N_2},\ldots,1/\sqrt{N_2}\}^T $ and $ \ket{\phi_{N_1}} = \{1/\sqrt{N_1},\ldots,1/\sqrt{N_1},0,\ldots,0\}^T $ as the reference states $\ket{\phi_{r_1}}$ and $\ket{\phi_{r_2}}$ respectively. Both reference states can be generated by integrating step functions, which are efficiently integrable. Lastly, condition 2(a) in Theorem \ref{thm:genres} is also satisfied by the definition of $Z_1$ and $Z_2$, since at most $ O(\mbox{log}(N_1+N_2)) $ repetitions of $U$ (with different sets of conditionals) are required to implement the controlled unitary operation in either case. Hence, by Theorem \ref{thm:genres}, $U_{walk}$ for a complete bipartite graph can be efficiently implemented.

\begin{figure}[htp]
	\centering
	\includegraphics[scale=0.35]{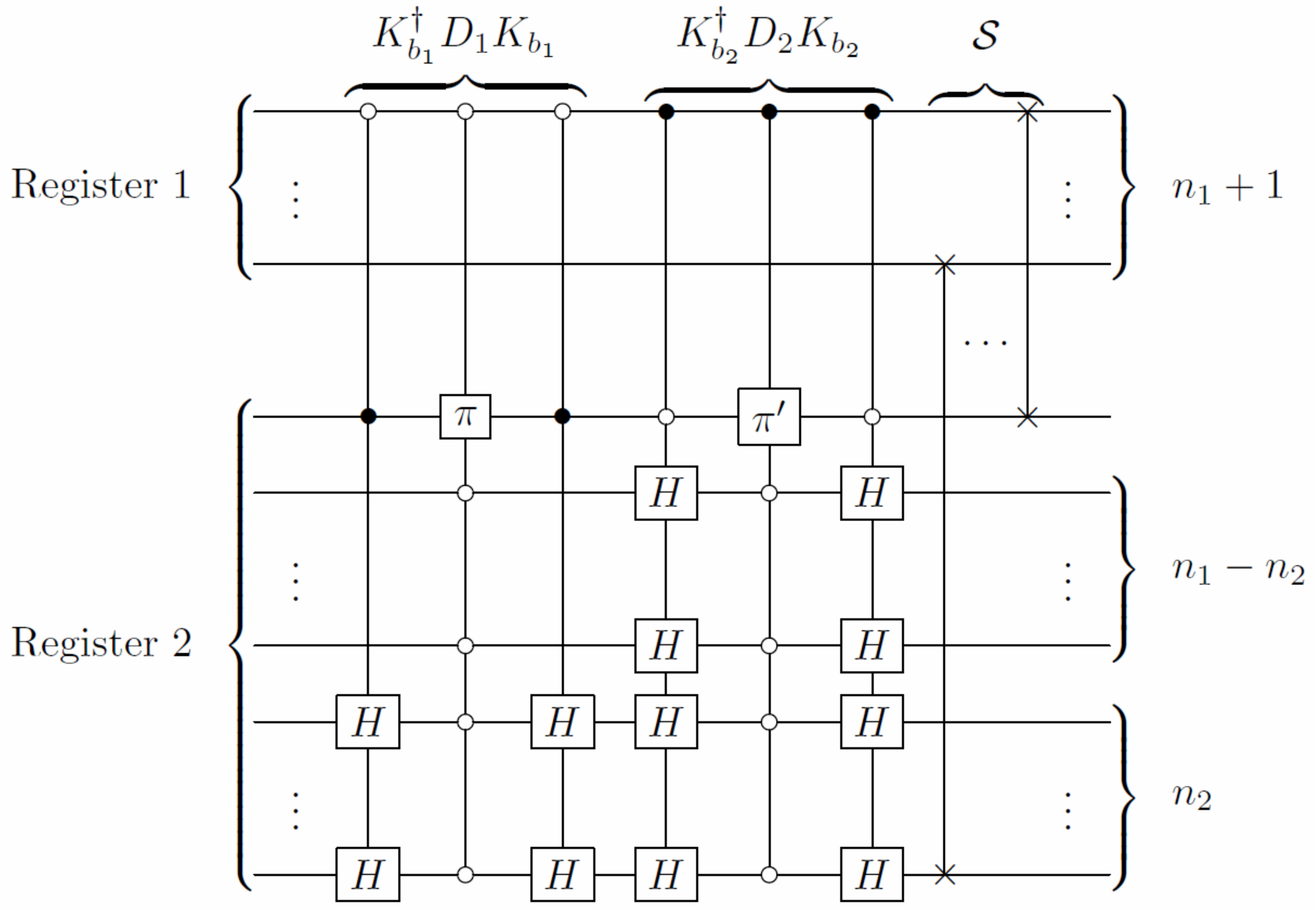}
	\caption{Quantum circuit implementing $U_{walk}$ for the $K_{2^{n_1},2^{n_2}}$ graph.}
	\label{fig:Kmnimpl}
\end{figure}

In the case where $N_1=N_2$, the complete bipartite graph is isomorphic to the circulant graph with $ \ket{\phi_0} = \{ 0, \sqrt{2/N_1}, 0, \sqrt{2/N_1}, ..., 0, \sqrt{2/N_1} \}^T $. For the particular case where $N_1 = 2^{n_1}$ and $N_2 = 2^{n_2}$ (where $n_1, n_2 \in \mathbb{Z}_+ $ and $n_1 \geq n_2$), the quantum circuit for the Szegedy walk operator $U_{walk}$ can be implemented explicitly without using the state preparation method using integrals, as shown in Figure \ref{fig:Kmnimpl}.

%More dense graphs, such as the half-complete graph, which is circulant with $ \ket{\phi_0} = \{ 0, \sqrt{2/N}, 0, \sqrt{2/N}, ..., 0, \sqrt{2/N} \}^T $, can be implemented in similar fashion using Hadamard and NOT gates for the preparation routine $K_b$. Note, however, that we are not constrained to work with transition matrices $P$ which correspond to a discrete-time quantum walk - any kind of circulant matrix where $ \ket{\phi_r} $ can be efficiently prepared can be implemented in this fashion.

\section{Tensor product of Markov chains}
\label{sec:tensor}

Here, we consider the composition of Markov chains to produce a new Markov chain. The most natural composition of Markov chains is the tensor product, i.e. given two Markov chains defined by the transition matrices $P_1$ and $P_2$, a new Markov chain can be defined by $P = P_1 \otimes P_2$. 

\begin{theorem}
	The Szegedy walk operator $U_{walk}$ corresponding to the Markov chain with transition matrix $P = P_1 \otimes P_2$ can be efficiently implemented if the Szegedy walk operators $U_{walk,1}$ and $U_{walk,2}$ corresponding to the Markov chains with transition matrix $P_1$ and $P_2$ respectively can be efficiently implemented.
	\label{thm:tensor}
\end{theorem}

\begin{proof}
	From equation (\ref{eqn:Uwalk}), if the walk operator can be efficiently implemented, it follows that the associated reflection operator $\mathcal{R}$ can also be efficiently implemented. From equations (\ref{eqn:psi}), (\ref{eqn:proj}) and (\ref{eqn:diag}), we have that:
	\begin{eqnarray}
		\displaystyle \mathcal{R}_1 &=& 2\left( \sum_{i=0}^{N_1-1} \ket{i}\bra{i} \otimes \ket{\phi_{i,1}}\bra{\phi_{i,1}} \right)-I_{N_1^2} \\
		U_1 \mathcal{R}_1 U_1^\dagger &=& 2\left( \sum_{i=0}^{N_1-1} \ket{i}\bra{i} \otimes \ket{b_1}\bra{b_1} \right)-I_{N_1^2}
	\end{eqnarray}
	\noindent where $\ket{\phi_{i,1}}$ is the square root of the $i$th column of $P_1$. Similarly:
	\begin{eqnarray}
		\displaystyle \mathcal{R}_2 &=& 2\left( \sum_{j=0}^{N_2-1} \ket{j}\bra{j} \otimes \ket{\phi_{j,2}}\bra{\phi_{j,2}} \right)-I_{N_2^2} \\
		\displaystyle U_2 \mathcal{R}_2 U_2^\dagger &=& 2\left( \sum_{j=0}^{N_2-1} \ket{j}\bra{j} \otimes \ket{b_2}\bra{b_2} \right)-I_{N_2^2}
	\end{eqnarray}
	\noindent where $\ket{\phi_{j,2}}$ is the square root of the $j$th column of $P_2$. For the composite system:
	\begin{equation}
		\displaystyle \mathcal{R} = 2\left( \sum_{k=0}^{N_1 N_2-1} \ket{k}\bra{k} \otimes \ket{\phi_{k}}\bra{\phi_{k}} \right)-I_{N_1^2 N_2^2},
	\end{equation}
	\noindent where $\ket{\phi_{k}}$ is the square root of the $k$th column of $P$. Writing $k = iN_2+j $, this becomes:
	\begin{eqnarray}
		\displaystyle \mathcal{R} &=& 2\left( \sum_{i=0}^{N_1-1}\sum_{j=0}^{N_2-1} \ket{i,j}\bra{i,j} \otimes \ket{\phi_{i,1},\phi_{j,2}}\bra{\phi_{i,1},\phi_{j,2}} \right)-I_{N_1^2 N_2^2} \\
		&=& 2\left( \sum_{i=0}^{N_1-1}\sum_{j=0}^{N_2-1} \ket{i}\bra{i} \otimes \ket{j}\bra{j} \otimes \ket{\phi_{i,1}}\bra{\phi_{i,1}} \otimes \ket{\phi_{j,2}}\bra{\phi_{j,2}} \right)-I_{N_1^2 N_2^2},
	\end{eqnarray}
	\noindent since $P = P_1 \otimes P_2$. Now, if we apply the diagonalizing operation $U_1$ to the $\ket{i}$ and $\ket{\phi_{i,1}}$ registers of $\mathcal{R}$:
	\begin{equation}
		\displaystyle U_1 \mathcal{R} U_1^\dagger = 2\left( \sum_{i=0}^{N_1-1}\sum_{j=0}^{N_2-1} \ket{i}\bra{i} \otimes \ket{j}\bra{j} \otimes \ket{b_1}\bra{b_1} \otimes \ket{\phi_{j,2}}\bra{\phi_{j,2}} \right)-I_{N_1^2 N_2^2}.
	\end{equation}
	Similarly, applying the diagonalizing operation $U_2$ to the $\ket{j}$ and $\ket{\phi_{j,2}}$ registers of $\mathcal{R}$ gives:
	\begin{equation}
		\displaystyle U_2 U_1 \mathcal{R} U_1^\dagger U_2^\dagger = 2\left( \sum_{i=0}^{N_1-1}\sum_{j=0}^{N_2-1} \ket{i}\bra{i} \otimes \ket{j}\bra{j} \otimes \ket{b_1}\bra{b_1} \otimes \ket{b_2}\bra{b_2} \right)-I_{N_1^2 N_2^2},
	\end{equation}
	\noindent which can be rewritten as:
	\begin{equation}
		\displaystyle U_2 U_1 \mathcal{R} U_1^\dagger U_2^\dagger = I_{N_1 N_2} \otimes ( 2 \ket{b} \bra{b} - I_{N_1 N_2} ) \equiv D,
	\end{equation}
	\noindent where $\ket{b} = \ket{b_1,b_2}$ and $D$ can be readily implemented using a controlled-$\pi$ gate (with multiple conditionals). Hence $ \mathcal{R} = U_1^\dagger U_2^\dagger D U_2 U_1 $ (and hence $U_{walk}$) can be efficiently implemented when $U_{walk,1}$ and $U_{walk,2}$ can be efficiently implemented, as required.
\end{proof}

\begin{figure}[htp]
	\centering
	\includegraphics[scale=0.35]{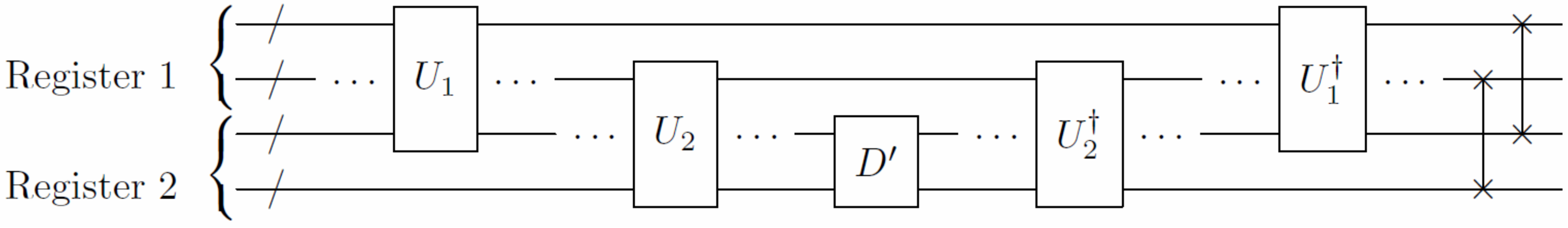}
	\caption{Quantum circuit implementing $U_{walk}$ corresponding to the Markov chain with transition matrix $P = P_1 \otimes P_2$, with $D' = 2 \ket{b_1,b_2}\bra{b_1,b_2} - I_{N_1 N_2} $.}
	\label{fig:Uimpl2}
\end{figure}

Figure \ref{fig:Uimpl2} shows the general quantum circuit for implementing $U_{walk}$ corresponding to the Markov chain with transition matrix $P = P_1 \otimes P_2$. The above result can also be easily extended to multiple systems as follows:

\begin{corollary}
	The Szegedy walk operator $U_{walk}$ corresponding to the Markov chain with transition matrix $\displaystyle P = P_1 \otimes \ldots \otimes P_p $ can be efficiently implemented if the Szegedy walk operators $U_{walk,1},\ldots,U_{walk,p}$ corresponding to the Markov chains with transition matrix $P_1,\ldots,P_p$ respectively can be efficiently implemented.
	\label{crl:tensorp}
\end{corollary}

\begin{proof}
	This follows easily by applying Theorem \ref{thm:tensor} recursively.
\end{proof}

In the case of Markov chains on graphs described by adjacency matrices, the tensor product of transition matrices correspond to taking the graph tensor product \cite{sampathkumar_tensor_1975} of the underlying graphs. One important class of graphs formed from a tensor product of graphs is the bipartite double cover of a given graph $G$. If the adjacency matrix of $G$ is given by $A$, then the bipartite double cover of $G$ is given by the adjacency matrix $A \otimes K_2$, where $ K_2 = \begin{pmatrix} 0&1 \\ 1&0 \end{pmatrix} $ \cite{brouwer_distance-regular_1989}. From Theorem \ref{thm:tensor}, it follows that in order to implement the Szegedy walk on the bipartite double cover of a given graph, we need to be able to implement the Szegedy walk on the $K_2$ graph. Figure \ref{fig:K2} shows the quantum circuit implementation of the Szegedy walk with transition matrix $P = K_2$.

\begin{figure}[htp]
	\centering
	\includegraphics[scale=0.15]{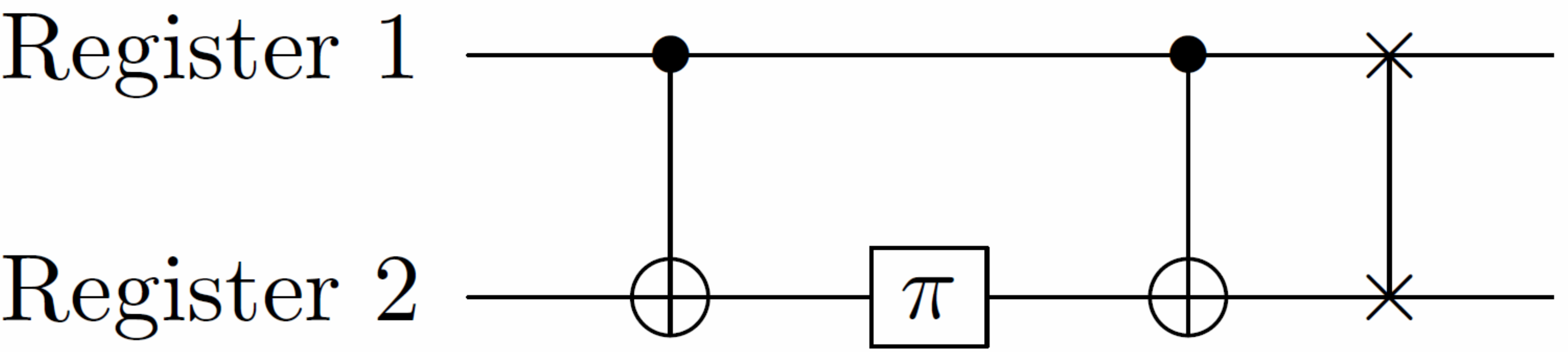}
	\caption{Quantum circuit implementing $U_{walk}$ for the $K_2$ graph.}
	\label{fig:K2}
\end{figure}

For example, the crown graph $S_N^0$ (example shown in Figure \ref{fig:Sn0}) can be formed by the bipartite double cover of the complete graph $K_N$ \cite{brouwer_distance-regular_1989}, i.e. $S_N^0 = K_N \otimes K_2$. From section \ref{sec:cyclic}, we already have an efficient implementation for the Szegedy walk operator for the $K_N$ graph (where $N=2^n$ for some $n \in \mathbb{Z}_+$)---specifically, refer to Figures \ref{fig:KnKb} and \ref{fig:Knimpl}. Hence, applying Theorem \ref{thm:tensor}, this means that we also have an efficient implementation for the crown graph $S_N^0$. By identifying $U_1$ and $U_2$ with the diagonalizing operations in Figures \ref{fig:Knimpl} and \ref{fig:K2} respectively, and taking $\ket{b_1} = \ket{0}$ and $\ket{b_2} = \ket{1}$ (giving $\ket{b_1,b_2}=\ket{0,1}$), we can implement the Szegedy walk operator $U_{walk}$ for the crown graph $S_N^0$, as shown in Figure \ref{fig:Sn0circ}.

\begin{figure}[htp]
	\centering
	\includegraphics[scale=0.30]{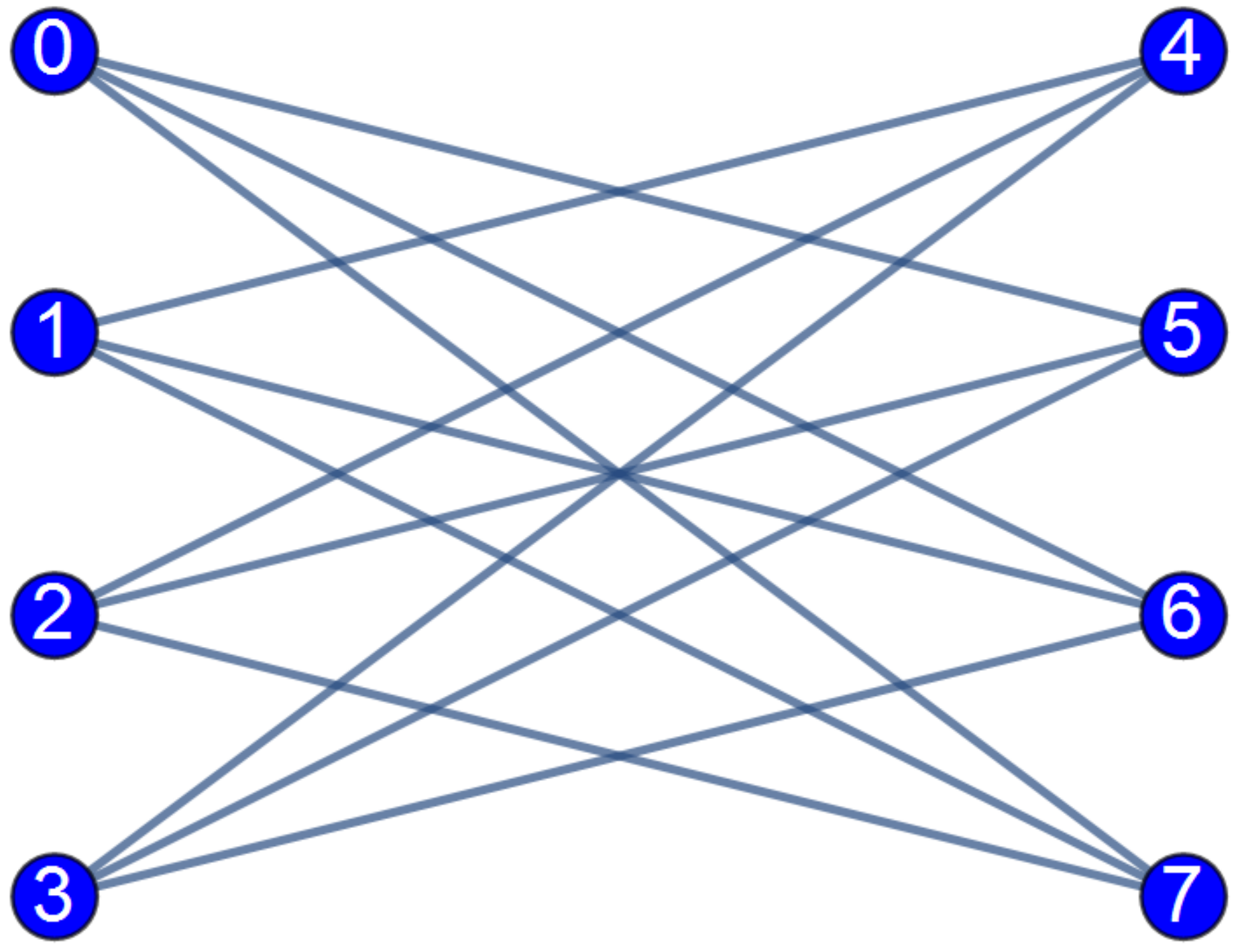}
	\caption{Crown graph $S_N^0$ with parameter $N=4$.}
	\label{fig:Sn0}
\end{figure}

\begin{figure}[htp]
	\centering
	\includegraphics[scale=0.32]{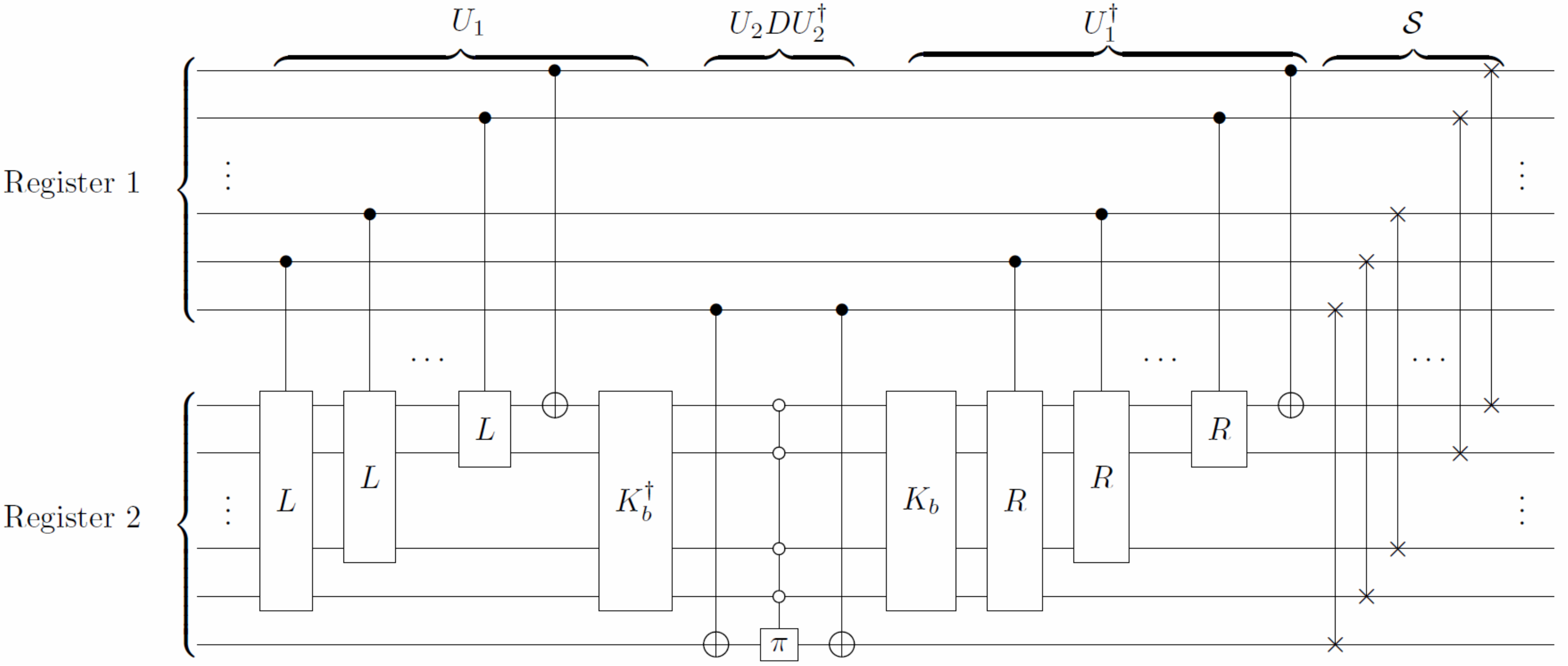}
	\caption{Quantum circuit implementing $U_{walk}$ for the $S_N^0$ graph, where the implementation of $K_b$ is given by Figure \ref{fig:KnKb}.}
	\label{fig:Sn0circ}
\end{figure}

\section{Weighted interdependent networks}
\label{sec:win}

A composition of Markov chains in the form of a weighted interconnected network (a generalization of the interdependent network defined in \cite{nehaniv_construction_2014}) can also be considered. Suppose we have the transition matrices $ P_A = \begin{pmatrix} A_1&0 \\ 0&A_2 \end{pmatrix} $ and $ P_B = \begin{pmatrix} 0&B_2 \\ B_1&0 \end{pmatrix} $ (where $A_1$, $A_2$, $B_1$, and $B_2$ have dimensions $N_1$-by-$N_1$, $N_2$-by-$N_2$, $N_2$-by-$N_1$, and $N_1$-by-$N_2$ respectively). Then a weighted interdependent network formed from $P_A$ and $P_B$ has the form:

\begin{equation}
P = \begin{pmatrix} \alpha_1 A_1& (1-\alpha_2)B_2 \\ (1-\alpha_2) B_1& \alpha_2 A_2 \end{pmatrix},
\end{equation}

\noindent where $\alpha_1,\alpha_2 \in [0,1]$ are independent real parameters. Note that this definition preserves column-normalization, assuming $P_A$ and $P_B$ are column-normalized. Under some circumstances, it is possible to show that the Szegedy walk operator corresponding to $P$ is efficiently implementable if the Szegedy walk operator corresponding to $P_A$ and $P_B$ are efficiently implementable. However, for this section, we will simply demonstrate the efficient implementation of a particular case.

\begin{figure}[htp]
	\centering
	\includegraphics[scale=0.25]{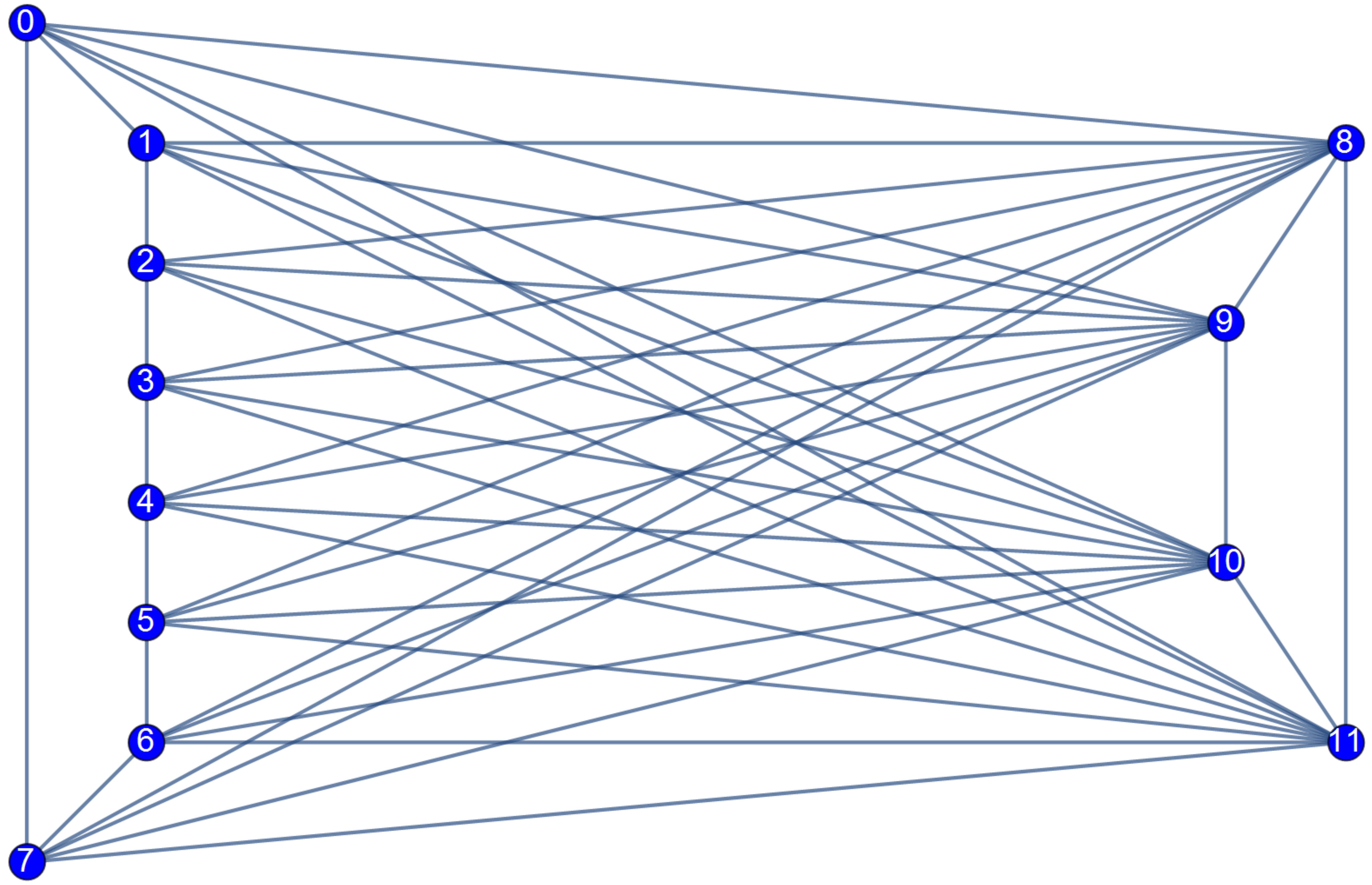}
	\caption{Weighted interdependent network with parameters $N_1 = 8$ and $N_2 = 4$.}
	\label{fig:CKg}
\end{figure}

Consider the weighted interdependent network (example shown in Figure \ref{fig:CKg}) formed by connecting two disjoint cycle graphs $C_{N_1}$ and $C_{N_2}$ using complete interconnections (i.e. a complete bipartite graph $K_{N_1,N_2}$). Mathematically, this corresponds to:

\begin{equation}
P_A = \begin{pmatrix} \frac{1}{2}C_{N_1}&0 \\ 0&\frac{1}{2}C_{N_2} \end{pmatrix} \mbox{ and  } P_B = \begin{pmatrix} 0&\frac{1}{N_2}J_{N_1,N_2} \\ \frac{1}{N_1}J_{N_2,N_1}&0 \end{pmatrix},
\end{equation} 

\noindent which, from sections \ref{sec:cyclic} and \ref{sec:kmn}, we have an efficient implementation for the Szegedy walk operator of each. We can combine these Markov chains to form a weighted interconnected network, giving the transition matrix:

\begin{equation}
P = \begin{pmatrix} \frac{1}{2+N_2}C_{N_1}&\frac{1}{2+N_1}J_{N_1,N_2} \\ \frac{1}{2+N_2}J_{N_2,N_1}&\frac{1}{2+N_1}C_{N_2} \end{pmatrix},
\end{equation}

\noindent where we have chosen the parameters $\alpha_1$ and $\alpha_2$ such that each non-zero element in a column is weighted equally. Set the partition as $Z = Z_1 \cup Z_2$ where $Z_1 = \{0,\ldots,N_1-1\}$ and $Z_2 = \{N_1,\ldots,N_1+N_2-1\}$. For the subsets $Z_1$ and $Z_2$, each column has $(2+N_2)$ and $(2+N_1)$ non-zero elements respectively. Select the reference states as $\ket{\phi_{r_1}}=\ket{\phi_0}$ and $\ket{\phi_{r_2}}=\ket{\phi_{N_1-1}}$, where $\ket{\phi_0}=\frac{1}{\sqrt{2+N_2}}\{0,1,0,\ldots,0,1,1,\ldots,1\}^T$ and $\ket{\phi_{N_1-1}}=\frac{1}{\sqrt{2+N_1}}\{1,\ldots,1,0,1,0,\ldots,0,1\}^T$. For $Z_1$, the transformations $ T_{1,y}: \ket{\phi_y} \rightarrow \ket{\phi_0} $ can be defined as a (restricted) cyclic permutation of the reference state $\ket{\phi_0}$, i.e. $ T_{1,y} = L^y $ over the first $N_1$ rows. For $Z_2$, the transformations $ T_{2,y}: \ket{\phi_y} \rightarrow \ket{\phi_{N_1-1}} $ can be defined in similar fashion as $ T_{2,y} = L^{y-(N_1-1)} $ over the last $N_2$ rows. In the case where $N_1 = 2^{n_1}$ and $N_2 = 2^{n_2}$ (where $n_1, n_2 \in \mathbb{Z}_+ $ and $n_1 \geq n_2$), the quantum circuit for the Szegedy walk operator $U_{walk}$ can be implemented explicitly, as shown in Figure \ref{fig:CK}

\begin{figure}[htp]
	\centering
	\subfigure[\mbox{ }$K_{b_{A_1}}$]{\includegraphics[scale=0.25]{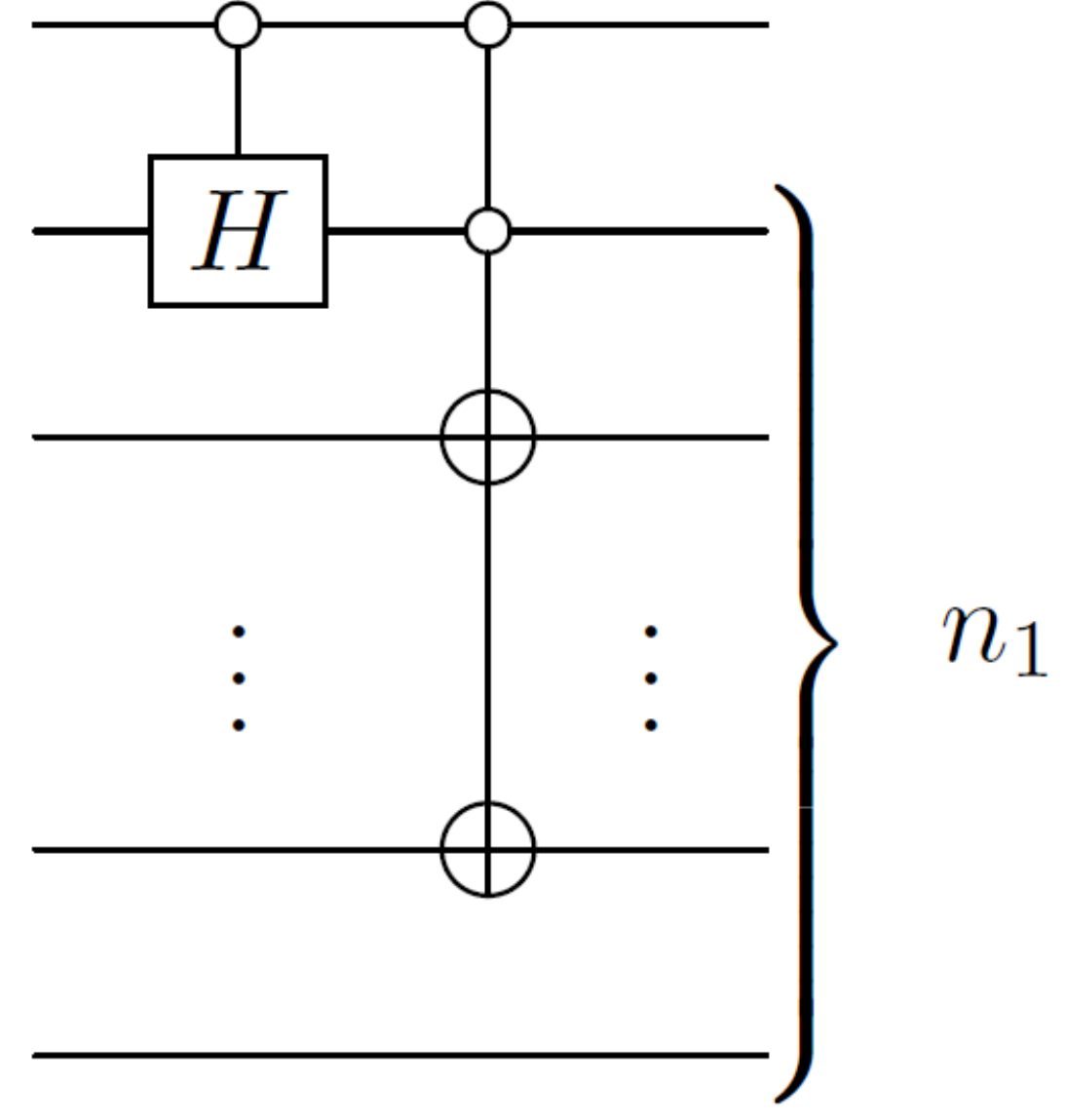}\qquad}
	\subfigure[\mbox{ }$K_{b_{B_1}}$]{\includegraphics[scale=0.25]{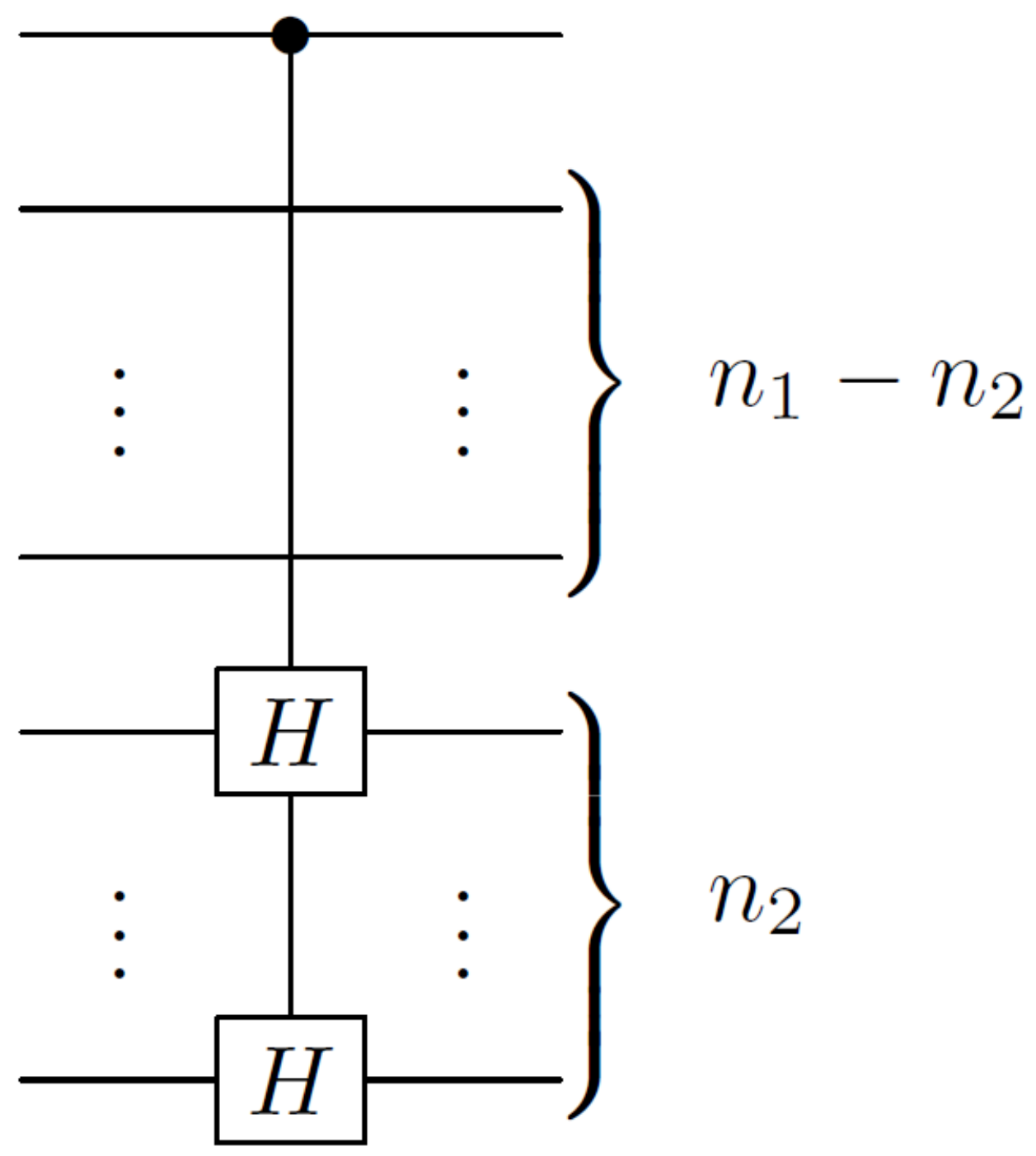}\qquad}
	\subfigure[\mbox{ }$K_{b_1}$]{\includegraphics[scale=0.25]{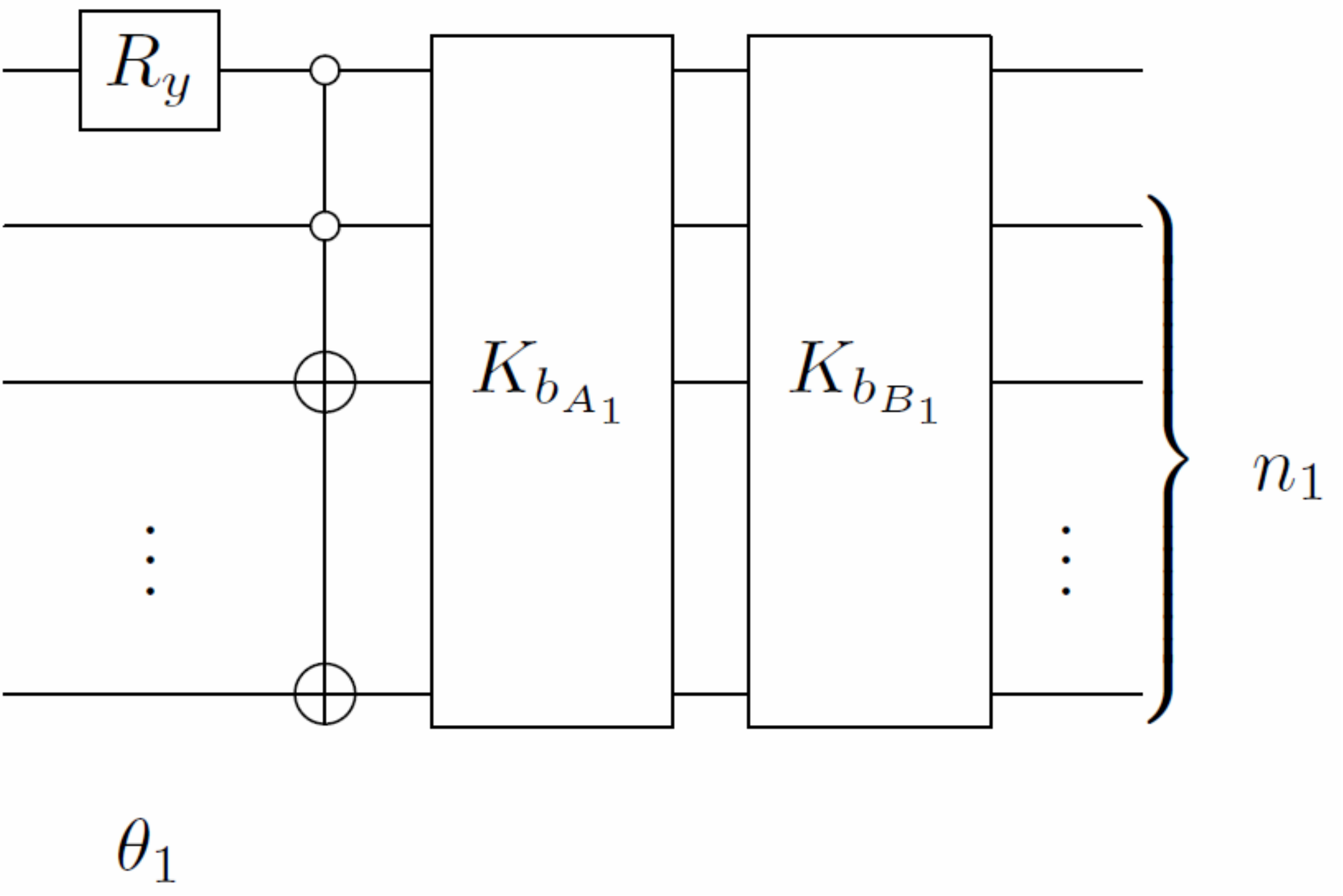}}
	\subfigure[\mbox{ }$K_{b_{B_2}}$]{\includegraphics[scale=0.25]{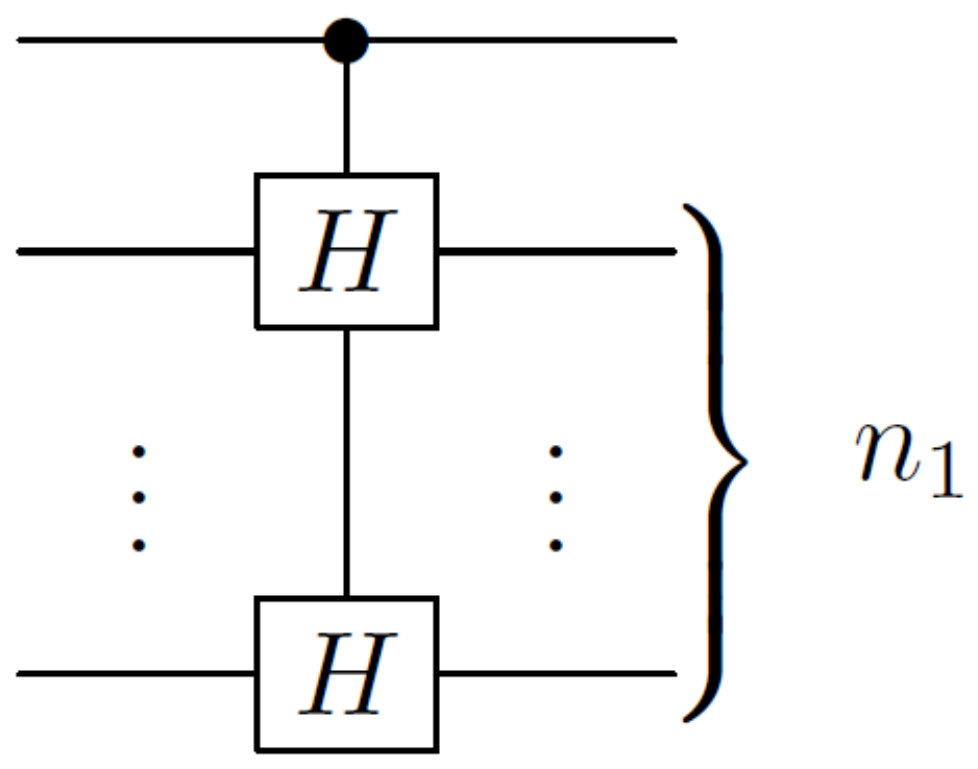}\qquad}
	\subfigure[\mbox{ }$K_{b_{A_2}}$]{\includegraphics[scale=0.25]{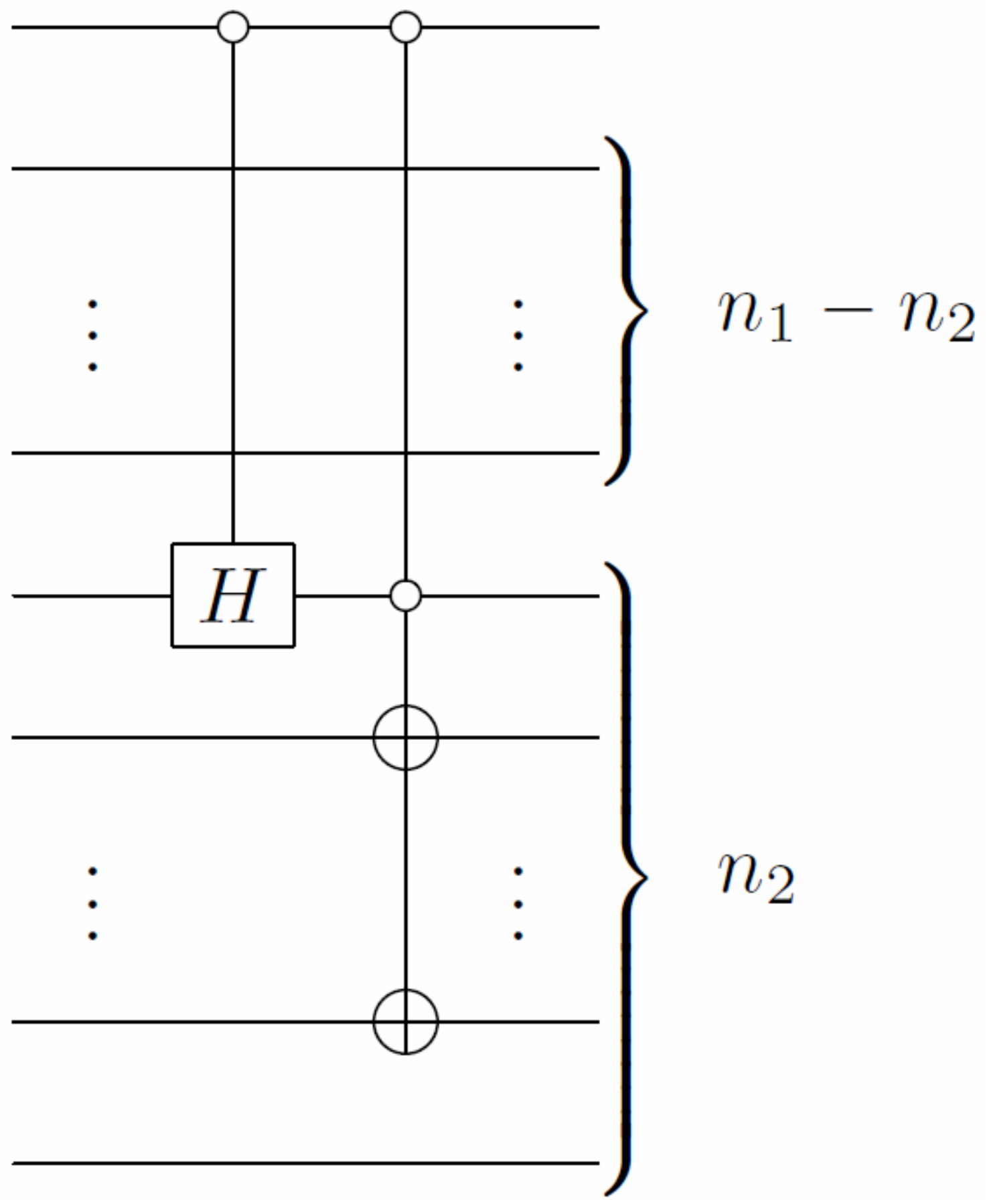}\qquad}
	\subfigure[\mbox{ }$K_{b_2}$]{\includegraphics[scale=0.25]{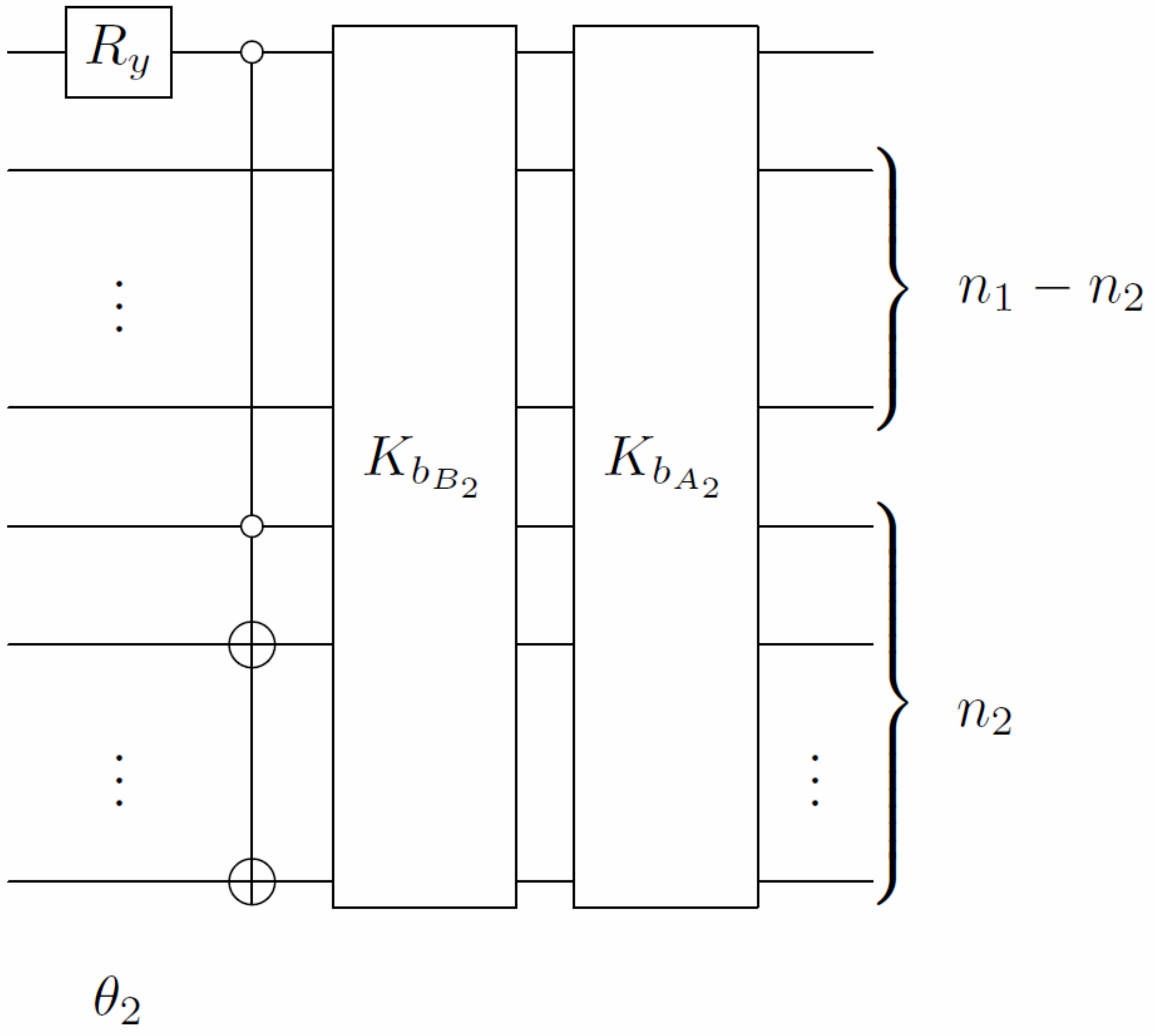}}
	\subfigure[Complete circuit for $ U_{walk} $]{\includegraphics[scale=0.32]{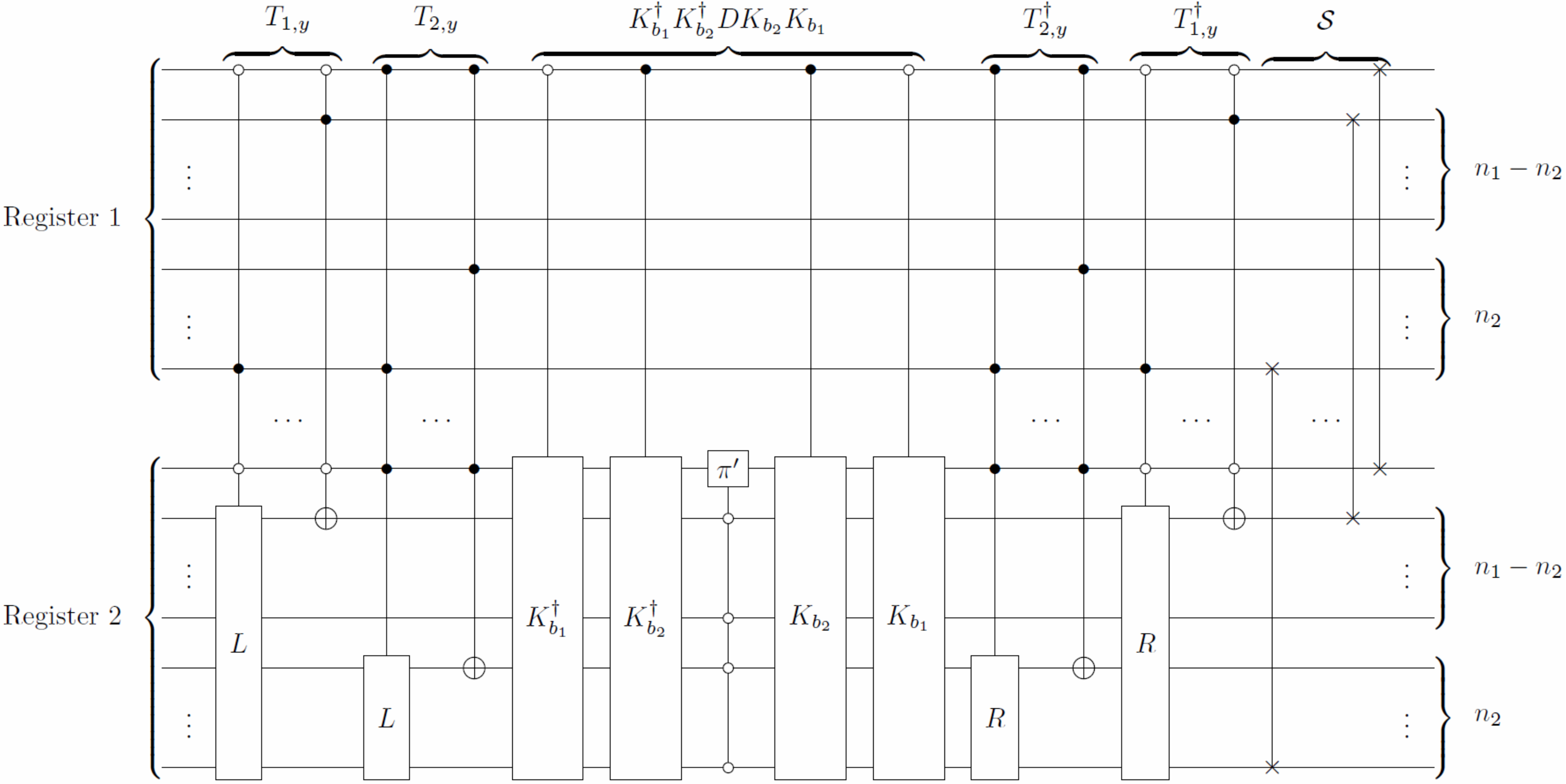}}
	\caption{The quantum circuit implementing $U_{walk}$ for the weighted interdependent network formed from disjoint cycle graphs connected using complete interconnections is shown in (g). The circuits implementing the preparation routines $K_{b_1} : \ket{0} \rightarrow \ket{\phi_0}$ and $K_{b_2} : \ket{0} \rightarrow \ket{\phi_{N_1-1}}$ are shown in (c) and (f) respectively, with rotation angles $\mbox{cos}(\theta_{1})=\sqrt{\frac{2}{2+N_2}}$ and $\mbox{cos}(\theta_{2})=\sqrt{\frac{N_1}{2+N_1}}$ respectively.}
	\label{fig:CK}
\end{figure}

\section{Application: Quantum Pagerank algorithm}
\label{sec:google}

The quantum Pagerank algorithm \cite{paparo_google_2012, paparo_quantum_2013} is a quantization of the classical Google Pagerank algorithm \cite{bryan_25000000000_2006} that is used to calculate the relative importance of nodes in a directed graph. It utilises Szegedy quantum walks to quantize the classical Markov chain, and the average probability of the wavefunction at a node is taken to be the measure of relative importance of that node.

Suppose that we have a directed graph (representing a network) described by a connectivity matrix $C$, where $C_{i,j} = 1$ if there is a link $j \rightarrow i$. Then, the patched connectivity matrix $E$ is defined as:

\begin{equation}
E_{i,j} = 
\begin{cases}
1/N & \mbox{outdeg}(j) = 0 \\
C_{i,j}/\mbox{outdeg}(j)  & \mbox{otherwise},
\end{cases}
\end{equation}

\noindent where $\mbox{outdeg}(j) = \displaystyle\sum_{i=1}^N C_{i,j}$ is the out-degree of vertex $j$. The Google matrix $G$ is then:

\begin{equation}
G = \alpha E + \frac{1-\alpha}{N}J,
\label{eqn:Gmat}
\end{equation}

\noindent where $\alpha \in [0,1]$ is the damping parameter (typically chosen to be $0.85$) and $J$ is the all-1s matrix. Taking the transition matrix as $P = G$, we can define the Szegedy walk operator $U_{walk}$ given by equation (\ref{eqn:Uwalk}). Then the instantaneous quantum Pagerank of the $j$th vertex is given by:

\begin{equation}
Q(j,t) = | \bra{j}_2 U_{walk}^{2t} \ket{\psi_0} |^2,
\end{equation}

\noindent where $\bra{j}_2 = (\ket{j}_2)^\dag$ and $\ket{j}_2$ is the $j$th standard basis vector of the second Hilbert space $\mathcal{H}_2$. The initial state $\ket{\psi_0}$ is taken to be an equal superposition over the $\ket{\psi_i}$ as defined in equation (\ref{eqn:psi}), i.e.:

\begin{equation} 
\ket{\psi_0} = \frac{1}{\sqrt{N}} \displaystyle\sum_{i=0}^{N-1} \ket{\psi_i}.
\end{equation}

The average quantum Pagerank for a vertex $j$, over some number of steps $T$, is defined as:

\begin{equation}
\langle Q(j) \rangle = \frac{1}{T} \displaystyle\sum_{t=0}^{T-1} Q(j,t),
\end{equation}

\noindent which can be shown to converge for sufficiently large $T$---and this is the quantity that is called the quantum Pagerank of a graph \cite{paparo_google_2012, paparo_quantum_2013}. Here, we are interested in simulating the walk operator $ U_{walk} $ using the Google matrix $G$ as the transition matrix, as this is necessary in order to compute the Google pagerank of the system. Note that from the definition of $G$ in equation (\ref{eqn:Gmat}), $G$ is not a sparse matrix whenever $\alpha \neq 1$, so it is not simulable by the methods of Chiang \emph{et al.} \cite{chiang_efficient_2010}.

Now, if $E$ is a matrix with columns related by cyclic permutations as defined before, then $G$ also has the same property, since addition by $\frac{1-\alpha}{N}\mathbf{1}$ does not change the cyclic permutation property, i.e. $T_i$ is the same as in section \ref{sec:cyclic} (but $K_b$ does change---see next paragraph). Of course, if all columns are related by cyclic permutations, then the Pagerank of each vertex would be equal---however we can exploit partitioning into subsets (from equations \ref{eqn:gen1}-\ref{eqn:gen4}) and cyclic symmetry to simulate the Szegedy walk using the Google matrix $G$ as the transition matrix on a graph that has non-equivalent sets of vertices.

In general, in order to prepare any column state $\ket{\phi_i}$ of the Google matrix $G$, we note that if the corresponding column state $\ket{\phi_i'}$ of the patched connectivity matrix $E$ can be prepared by the state preparation method using integrals (i.e. there exists an efficiently integrable function $p_i'(x)$ that generates the probability distribution $\left\{E_{j,i}\right\}$), then $p_i(x) = \alpha p_i'(x) + \frac{1-\alpha}{N L}$ (where $L$ is the length of the domain of $x$) is also an efficiently integrable function that generates $\left\{G_{j,i}\right\}$, i.e. the required column state $\ket{\phi_i}$. Hence if Corollary \ref{crl:mainres} can be applied to the patched connectivity matrix $E$, then it follows that it can also be applied to the Google matrix $G$ by changing the probability density function as above.

\begin{figure}[htp]
	\centering
	\subfigure[Undirected]{\label{fig:W8undirected} \includegraphics[scale=0.20]{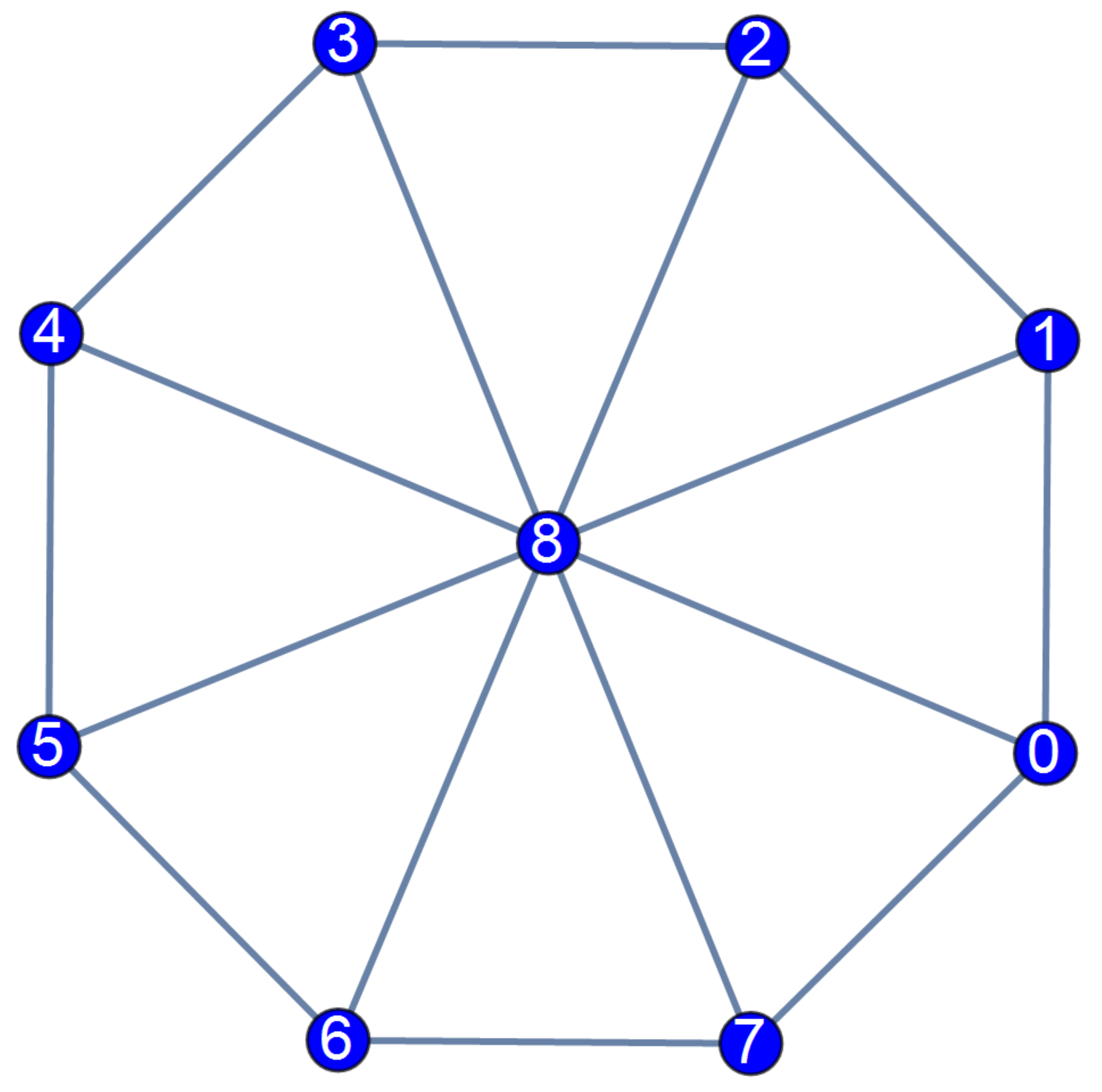}\qquad}
	\subfigure[Directed]{\label{fig:W8directed} \includegraphics[scale=0.20]{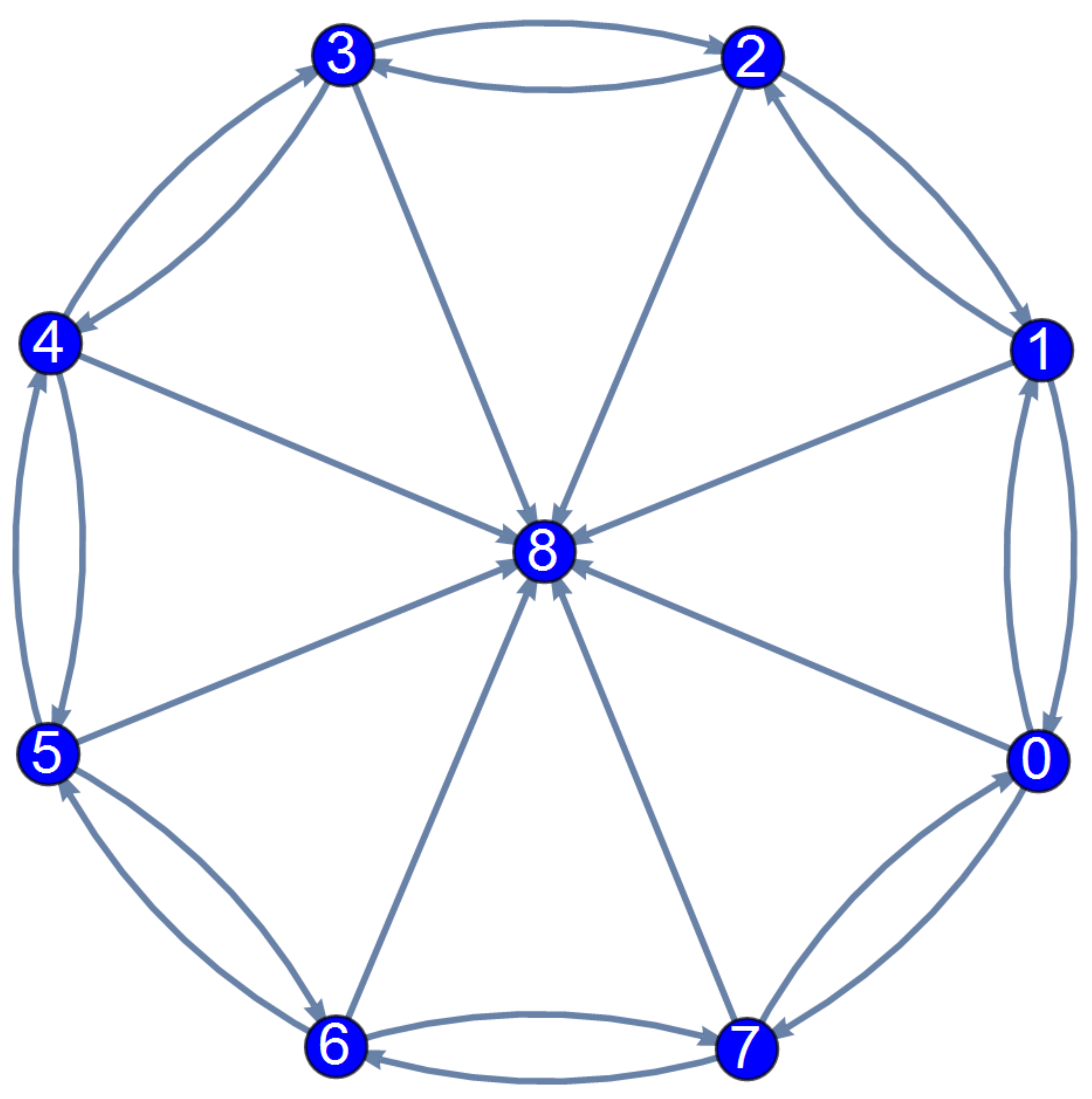}}
	\caption{Undirected wheel graph $W_N$ and its directed variant $W_N'$ with parameter $N=8$.}
	\label{fig:W8}
\end{figure}

As before, in some special cases, we can construct the preparation routine $K_b$ explicitly without using the state preparation method using integrals. One class of graphs for which this can be done is the wheel graph $W_N$ (undirected and directed example shown in Figure \ref{fig:W8}) is a graph that contains a cycle graph $C_N$ of length $N$, which has each vertex connected to a single hub vertex. The connectivity matrix for the undirected case shown in Figure \ref{fig:W8undirected} is given by the block matrix:

\begin{equation}
C(W_N) = \left(
\begin{array}{c c}
C(C_N) & 1 \\
1 & 0
\end{array}
\right),
\end{equation}

\noindent and the directed case shown in Figure \ref{fig:W8directed} is given by the modified block matrix:

\begin{equation}
C(W_N') = \left(
\begin{array}{c c}
C(C_N) & 0 \\
1 & 0
\end{array}
\right).
\end{equation}

We consider $N = 2^m$ for some $m\in \mathbb{Z}_+$. For both $ W_N $ and $W_N'$, the vertices can be partitioned into the sets $Z_1 = \{0,\ldots,N-1\}$ and $Z_2 = \{N\}$. For the set $Z_1$, we pick the reference state as $\ket{\phi_0}$, which from equation (\ref{eqn:Gmat}), can be written as $ \ket{\phi_0} = \{ \sqrt{\beta}, \sqrt{\gamma}, \sqrt{\beta}, \ldots, \sqrt{\beta}, \sqrt{\gamma}, \sqrt{\gamma} \}$ where $ \beta = \frac{(1-\alpha)}{N+1} $ and $ \gamma = \frac{\alpha}{3} + \beta $. Setting $\ket{b_1} = \ket{0}$, we can use the circuit in Figure \ref{fig:GKb1impl} to perform the operation $K_{b_1}:\ket{b_1} \rightarrow \ket{\phi_0}$.

The transformations $ T_{1,y}: \ket{\phi_y} \rightarrow \ket{\phi_0} $ can be defined as a (restricted) cyclic permutation of the reference state $\ket{\phi_0}$, i.e. $ T_{1,y} = L^y $. For the set $Z_2$, there is only one state, so $ T_{2,y} $ is not needed, and $K_{b_2}:\ket{b_2} \rightarrow \ket{\phi_{N}}$ can be constructed easily. Figure \ref{fig:Wnimpl} shows the complete circuit that simulates the Szegedy walk for the wheel graph $W_N$, the total cost being within $O(\mbox{poly}(\mbox{log}(N)))$ elementary gates.

\begin{figure}[htp]
	\centering
	\includegraphics[scale=0.30]{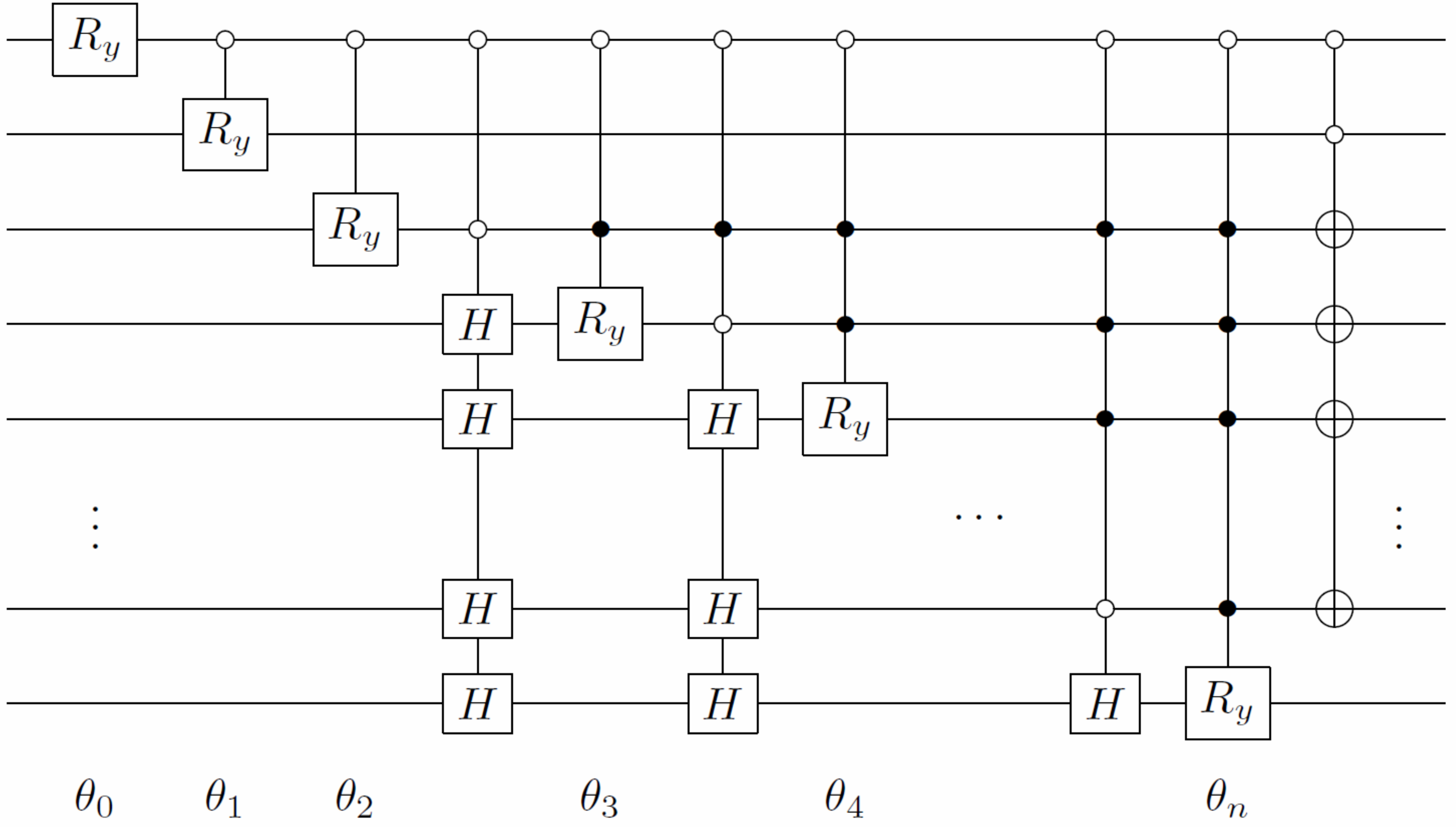}
	\caption{Quantum circuit implementing $K_{b_1}:\ket{b_1} \rightarrow \ket{\phi_0}$  for the vertex group $Z_1$ in the wheel graph $W_{2^m}$. The rotation angles are given by $\mbox{cos}(\theta_0) = \sqrt{1 - \gamma}$, $\mbox{cos}(\theta_1) = \sqrt{\frac{1}{2}}$ and $\mbox{cos}(\theta_{i>1}) = \sqrt{\frac{2^{m-i}\beta}{(2^{m-i+1}-1)\beta+\gamma}}$.}
	\label{fig:GKb1impl}
\end{figure}

\begin{figure}[htp]
	\centering
	\includegraphics[scale=0.35]{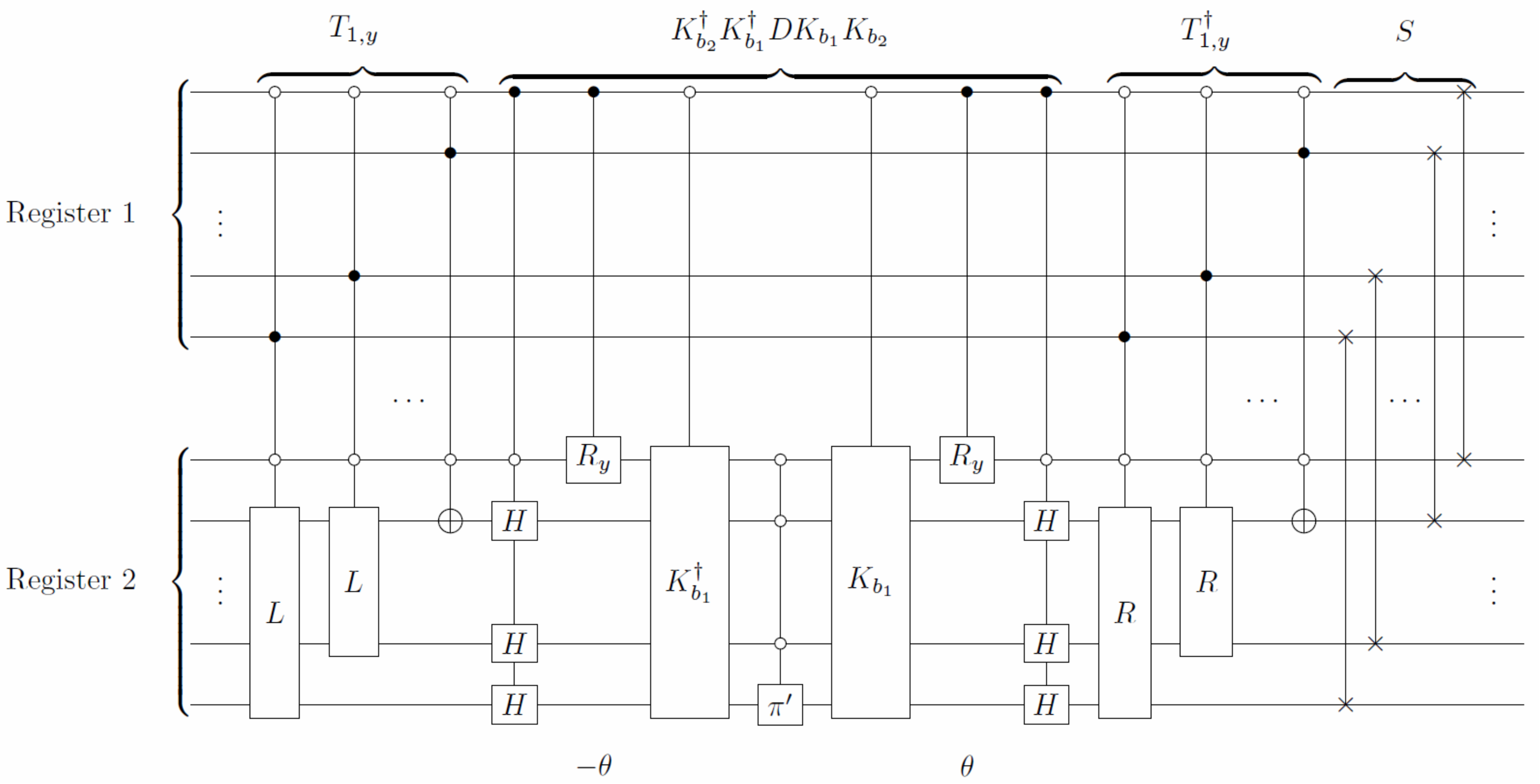}
	\caption{Quantum circuit implementing $U_{walk}$ for the $W_N$ and $W_N'$ graph. The rotation angle for $K_{b_2}$ is given by $ \mbox{cos}(\theta) = \sqrt{1-\beta} $ for the $W_N$ graph and $ \mbox{cos}(\theta) = \sqrt{\frac{N}{N+1}} $ for the $W_N'$ graph.}
	\label{fig:Wnimpl}
\end{figure}

Running the Szegedy walk using $U_{walk}^2$ as the walk operator, we obtain the instantaneous and average quantum Pagerank results, shown in Figure \ref{fig:W8res}. As expected, the hub vertex (corresponding to the set $Z_2$) has a much higher quantum Pagerank value than the outer vertices (corresponding to the set $Z_1$), which all have the same quantum Pagerank value since they are equivalent vertices. In the case of the directed graph $W_8'$, the hub vertex has a slightly higher quantum Pagerank value compared to its value in the case of the undirected graph $W_8$, because of the lack of outgoing edges from the hub in $W_8'$.

\begin{figure}[htp]
	\centering
	\subfigure[Instantaneous quantum Pagerank for $W_8$]{\includegraphics[scale=0.3]{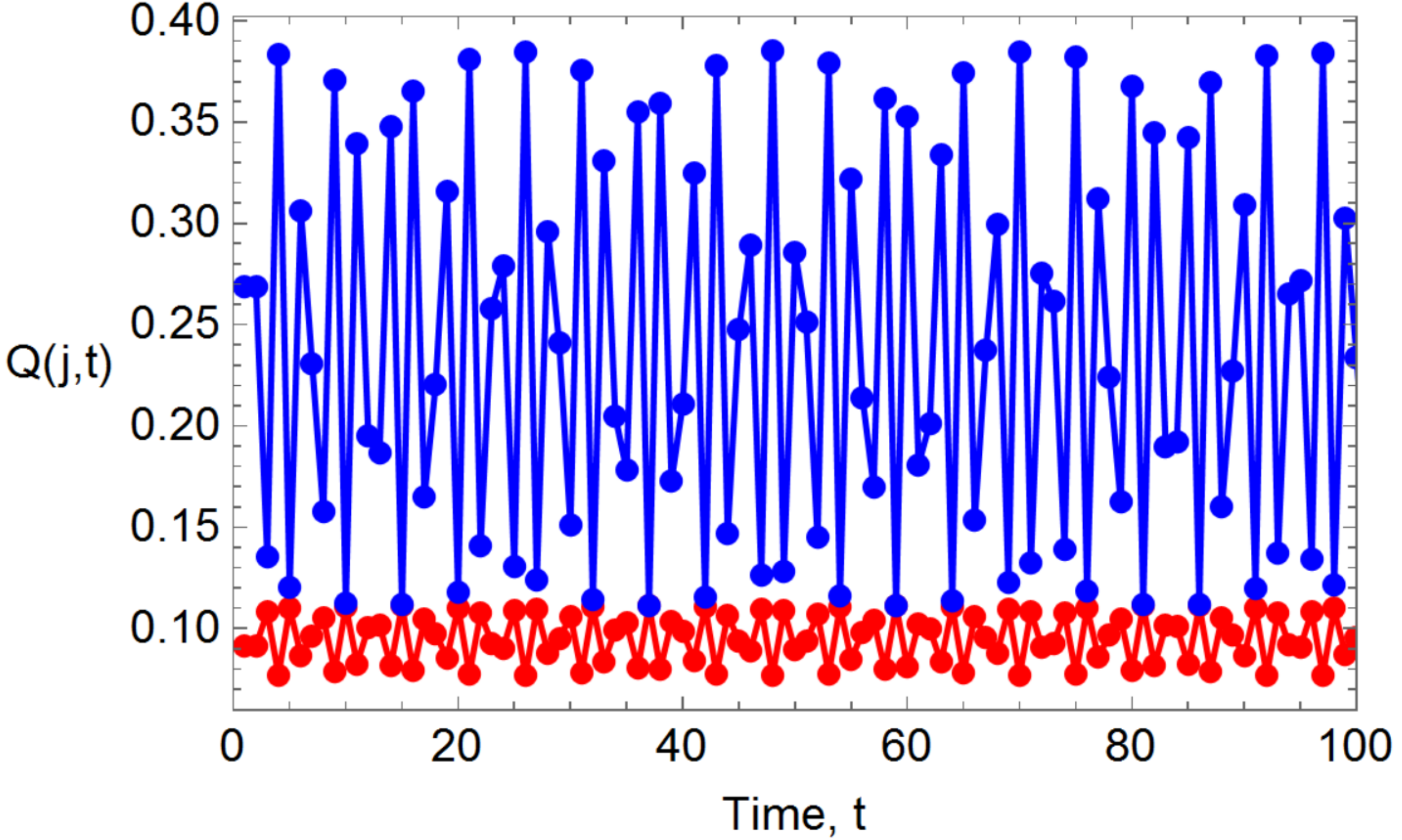}}
	\subfigure[Instantaneous quantum Pagerank for $W_8'$]{\includegraphics[scale=0.31]{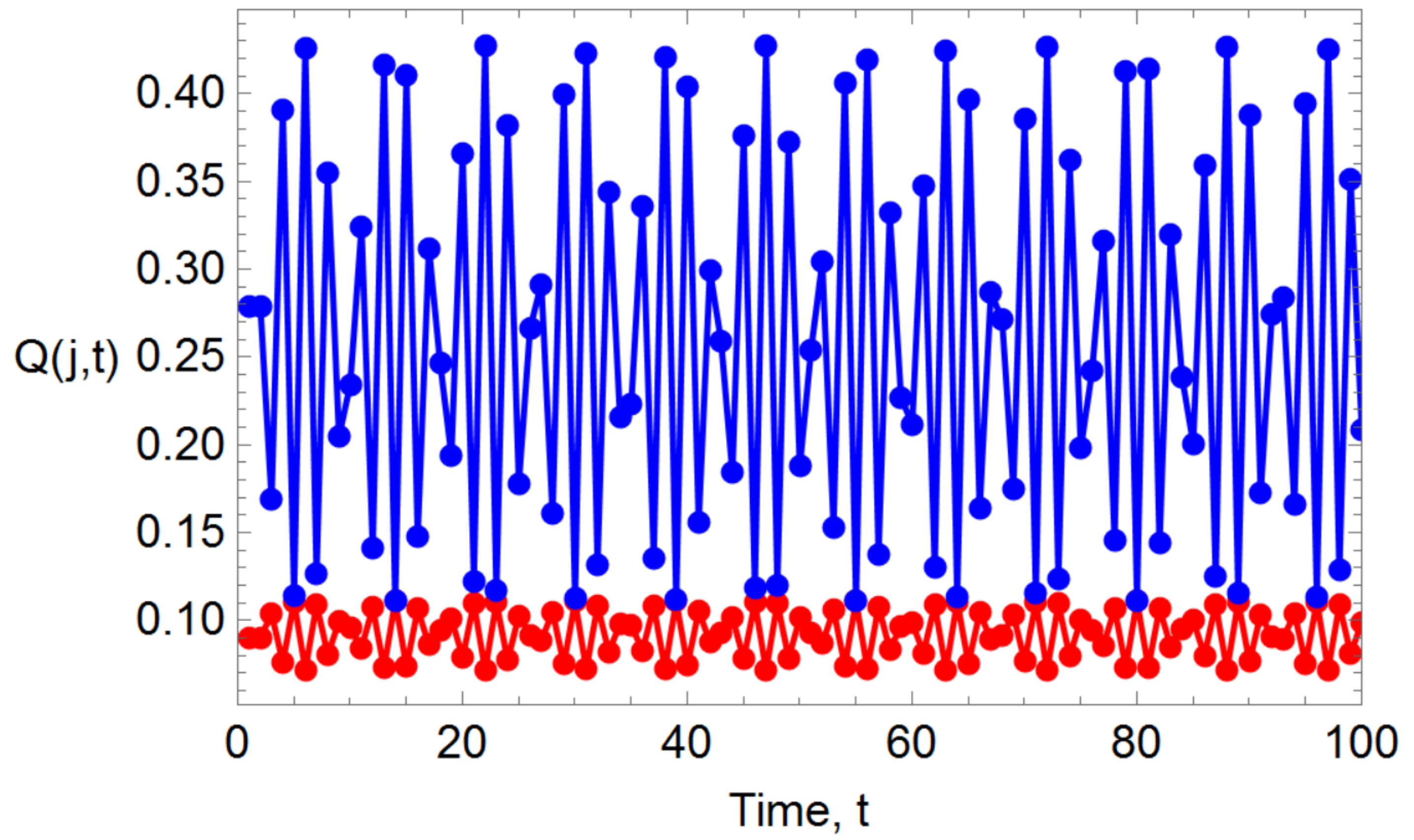}}
	\subfigure[Average quantum Pagerank]{\includegraphics[scale=0.29]{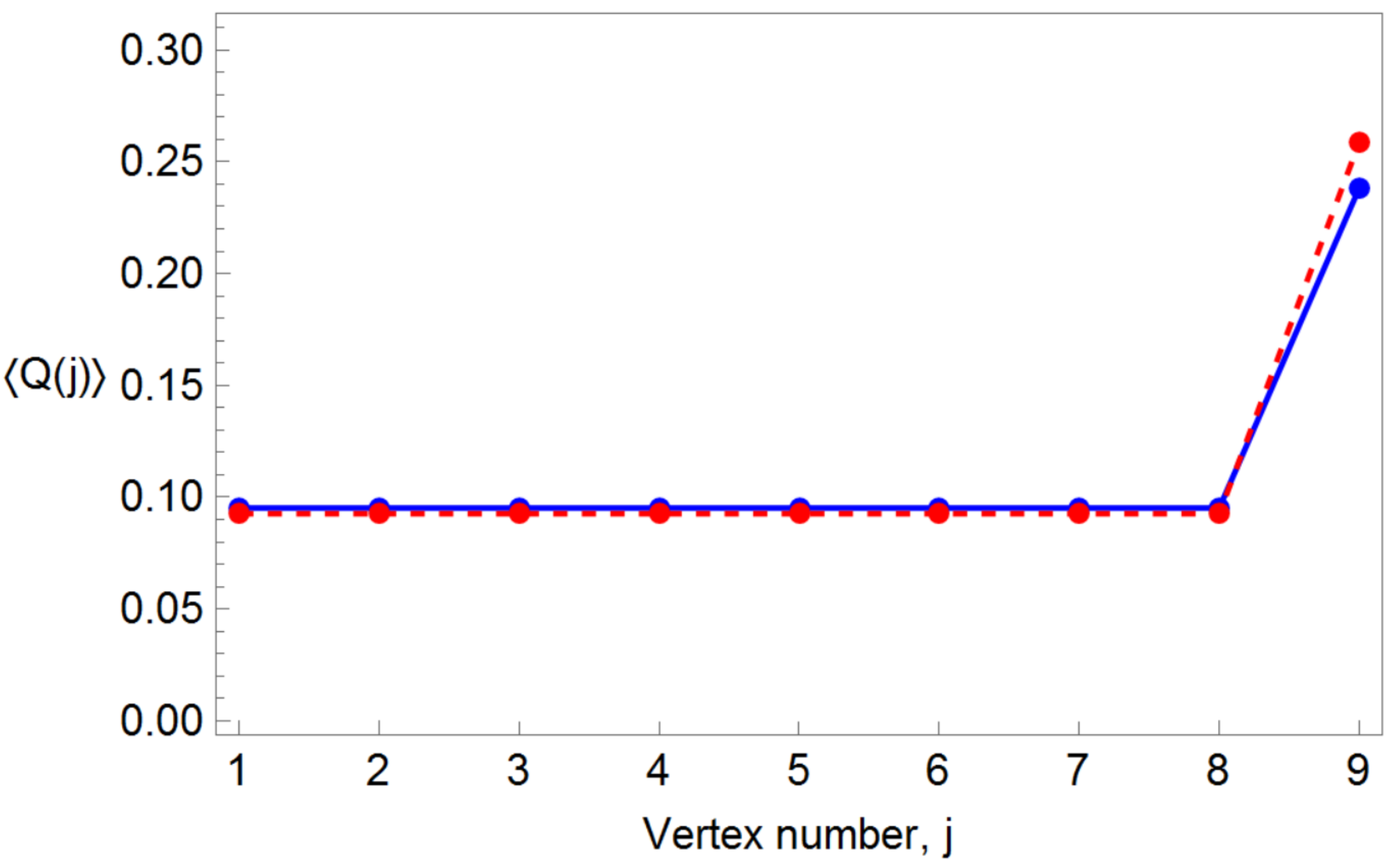}}
	\caption{Instantaneous and average quantum Pagerank results for the wheel graph $W_{8}$. In (a) and (b), the small red curve and large blue curve denotes the instantaneous quantum Pagerank for vertices 1-8 and 9 respectively (as labelled in Figure \ref{fig:W8}). In (c), the solid blue curve and the dashed red curve denotes the averaged quantum Pagerank for $W_8$ and $W_8'$ respectively.}
	\label{fig:W8res}
\end{figure}

\begin{figure}[htp]
	\centering
	\includegraphics[scale=0.25]{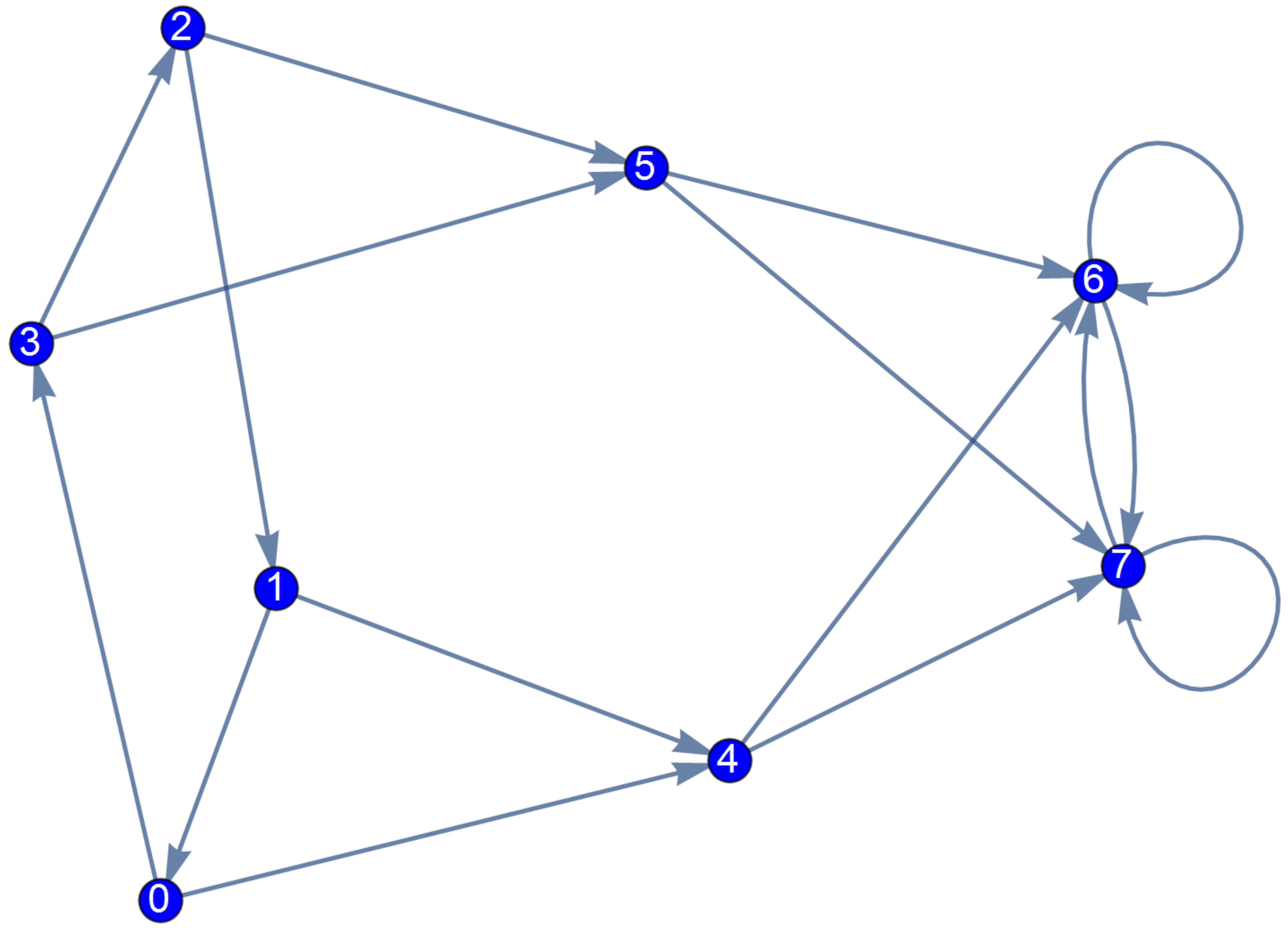}
	\caption{A directed graph on 8 vertices with three subsets of equivalent vertices.}
	\label{fig:Gdirected}
\end{figure}

A second example of computing the quantum Pagerank on directed graphs is shown in Figure \ref{fig:Gdirected}, the vertices of which can be partitioned into subsets of equivalent vertices as $Z = Z_1 \cup Z_2 \cup Z_3$ where $Z_1 = \{ 0,1,2,3 \}$, $Z_2 = \{ 4,5 \}$ and $Z_3 = \{ 6,7 \}$. The connectivity matrix is given by:

\begin{equation}
C = \left(
\begin{array}{c c c c c c c c}
0 & 0 & 0 & 1 & 1 & 0 & 0 & 0 \\
1 & 0 & 0 & 0 & 1 & 0 & 0 & 0 \\
0 & 1 & 0 & 0 & 0 & 1 & 0 & 0 \\
0 & 0 & 1 & 0 & 0 & 1 & 0 & 0 \\
0 & 0 & 0 & 0 & 0 & 0 & 1 & 1 \\
0 & 0 & 0 & 0 & 0 & 0 & 1 & 1 \\
0 & 0 & 0 & 0 & 0 & 0 & 1 & 1 \\
0 & 0 & 0 & 0 & 0 & 0 & 1 & 1
\end{array}
\right)
\end{equation}

Set the basis state for each set to be $ \ket{b_1} = \ket{b_2} = \ket{b_3} = \ket{0} $. For the set $Z_1$, we pick the reference state as $\ket{\phi_0}$, which can be written as $\ket{\phi_0} = \{ \sqrt{\beta}, \sqrt{\gamma_1}, \sqrt{\beta}, \sqrt{\beta}, \sqrt{\beta}, \sqrt{\beta}, \sqrt{\beta}, \sqrt{\beta} \}$ where $\beta = \frac{1-\alpha}{8}$ and $\gamma_1 = \alpha + \beta$. The required transformations $T_{1,y}$ for $y \in Z_1$ that does $ T_{1,y}:\ket{\phi_y} \rightarrow \ket{\phi_0} $ can be identified as $ T_{1,y} = L^{y} $. For the set $Z_2$, we pick the reference state as $\ket{\phi_4}$, which can be written as $\ket{\phi_4} = \{ \sqrt{\gamma_2}, \sqrt{\gamma_2}, \sqrt{\beta}, \sqrt{\beta}, \sqrt{\beta}, \sqrt{\beta}, \sqrt{\beta}, \sqrt{\beta} \}$ where $\gamma_2 = \frac{\alpha}{2} + \beta$. The required transformations $T_{2,y}$ for $y \in Z_2$ that does the analogous transformation is simply $ T_{2,4} = I $ and $ T_{2,5} = L^2 $. For the set $Z_3$, we pick the reference state as $\ket{\phi_6}$, which can be written as $\ket{\phi_6} = \{ \sqrt{\beta}, \sqrt{\beta}, \sqrt{\beta}, \sqrt{\beta}, \sqrt{\gamma_3}, \sqrt{\gamma_3}, \sqrt{\gamma_3}, \sqrt{\gamma_3} \}$ where $\gamma_3 = \frac{\alpha}{4} + \beta$. Since $ \ket{\phi_6} = \ket{\phi_7} $, no transformations are required. Figure \ref{fig:Gcircall} shows the quantum circuit implementing $ U_{walk} $ for this directed graph.

\begin{figure}[htp]
	\centering
	\subfigure[\mbox{ }$K_{b_1}$]{\includegraphics[scale=0.17]{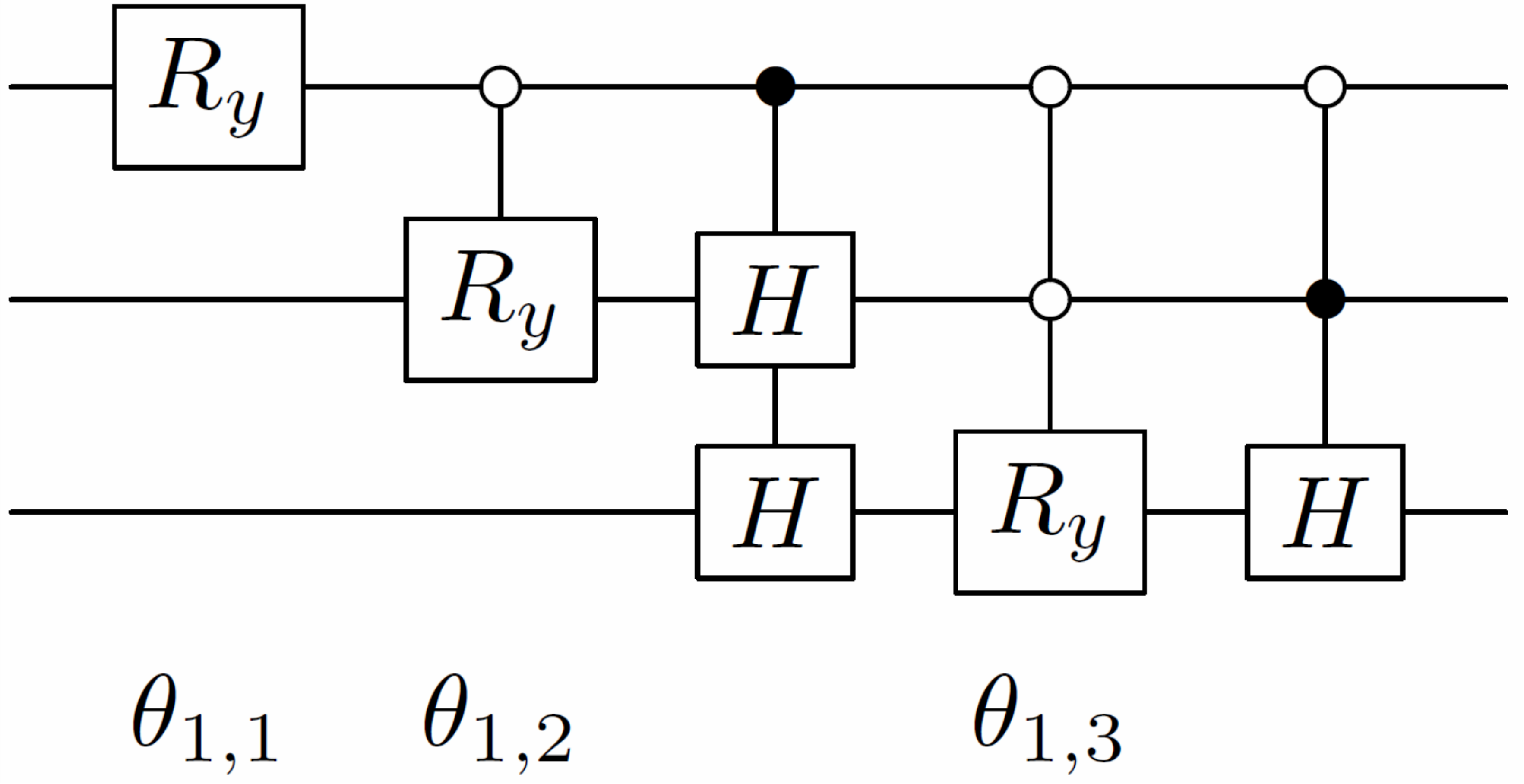}\qquad}
	\subfigure[\mbox{ }$K_{b_2}$]{\includegraphics[scale=0.17]{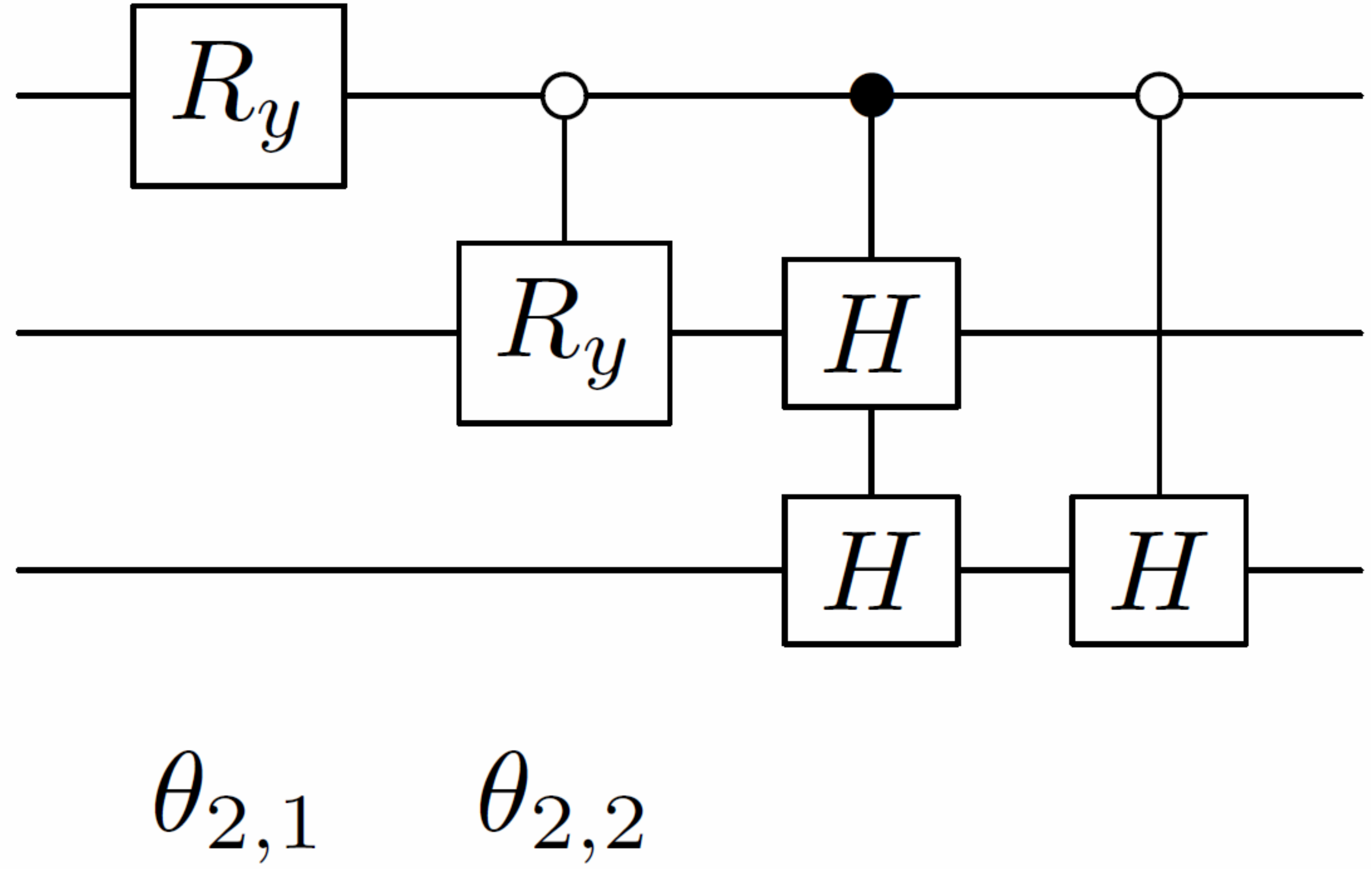}\qquad}
	\subfigure[\mbox{ }$K_{b_3}$]{\includegraphics[scale=0.17]{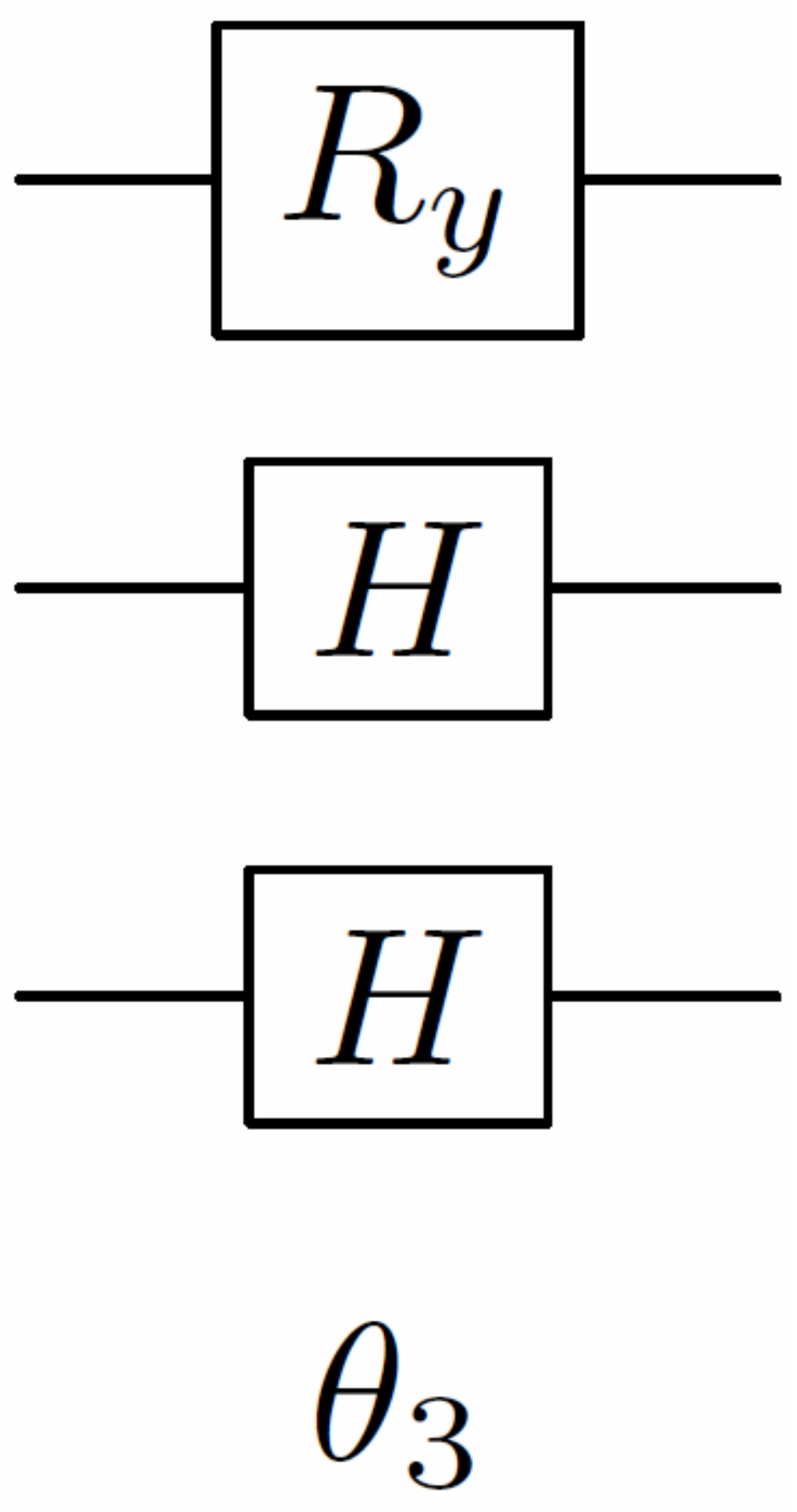}}
	\subfigure[Complete circuit for $ U_{walk} $]{\includegraphics[scale=0.30]{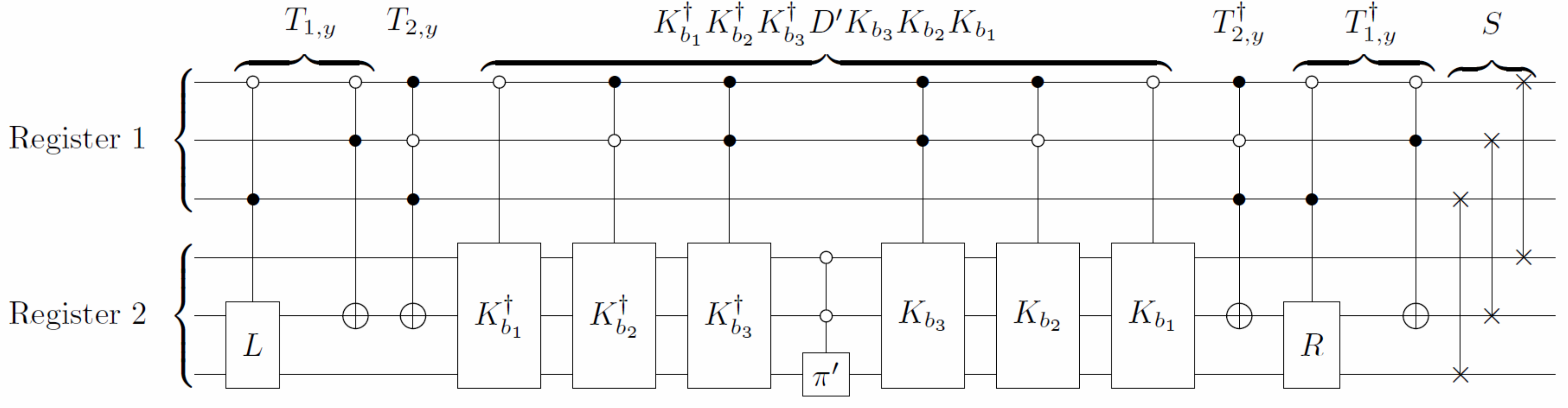}}
	\caption{Quantum circuit implementation for $U_{walk}$ shown in (d), with circuits implementing the preparation routines $K_{b_1} : \ket{0} \rightarrow \ket{\phi_0}$, $K_{b_2} : \ket{0} \rightarrow \ket{\phi_4}$ and $K_{b_3} : \ket{0} \rightarrow \ket{\phi_6}$ in (a), (b) and (c) respectively. The rotation angles used in (a) to (c) are $\mbox{cos}(\theta_{1,1})=\sqrt{\frac{3\beta+\gamma_1}{7\beta+\gamma_1}}$, $\mbox{cos}(\theta_{1,2})=\sqrt{\frac{\beta+\gamma_1}{3\beta+\gamma_1}}$, $\mbox{cos}(\theta_{1,3})=\sqrt{\frac{\beta}{\beta+\gamma_1}}$, $\mbox{cos}(\theta_{2,1})=\sqrt{\frac{\beta+\gamma_2}{3\beta+\gamma_2}}$, $\mbox{cos}(\theta_{2,2})=\sqrt{\frac{\gamma_2}{\beta+\gamma_2}}$ and $\mbox{cos}(\theta_3)=\sqrt{\frac{\beta}{\beta+\gamma_3}}$.}
	\label{fig:Gcircall}
\end{figure}

Running the Szegedy walk using $U_{walk}^2$ as the walk operator, we obtain the instantaneous and average quantum Pagerank results, shown in Figure \ref{fig:Gres}. We find the centrality of each subset to be ordered (from highest to lowest) as $Z_3$, $Z_2$ and $Z_1$.

\begin{figure}[htp]
	\centering
	\subfigure[Instantaneous quantum Pagerank]{\includegraphics[scale=0.3]{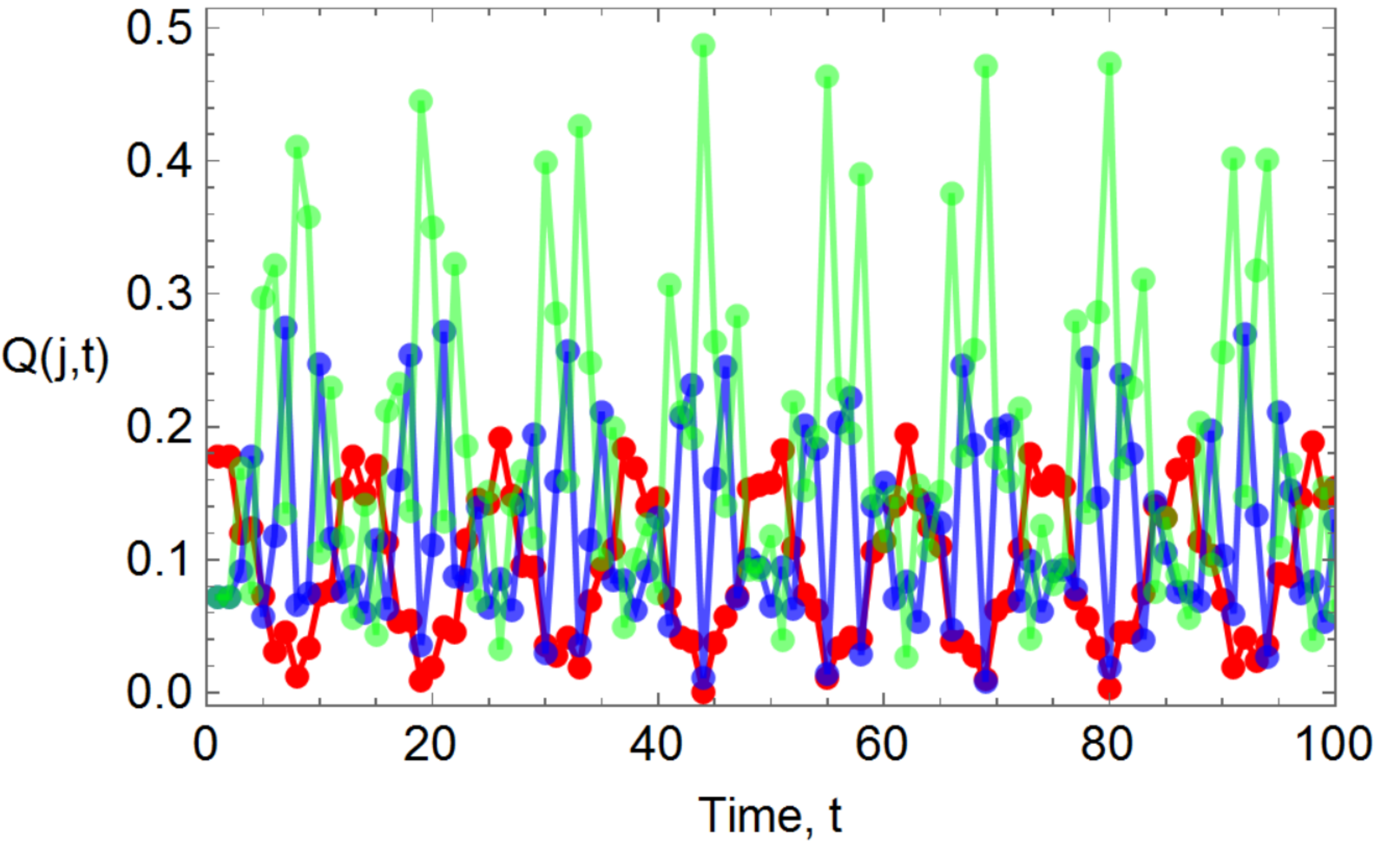}}
	\subfigure[Average quantum Pagerank]{\includegraphics[scale=0.28]{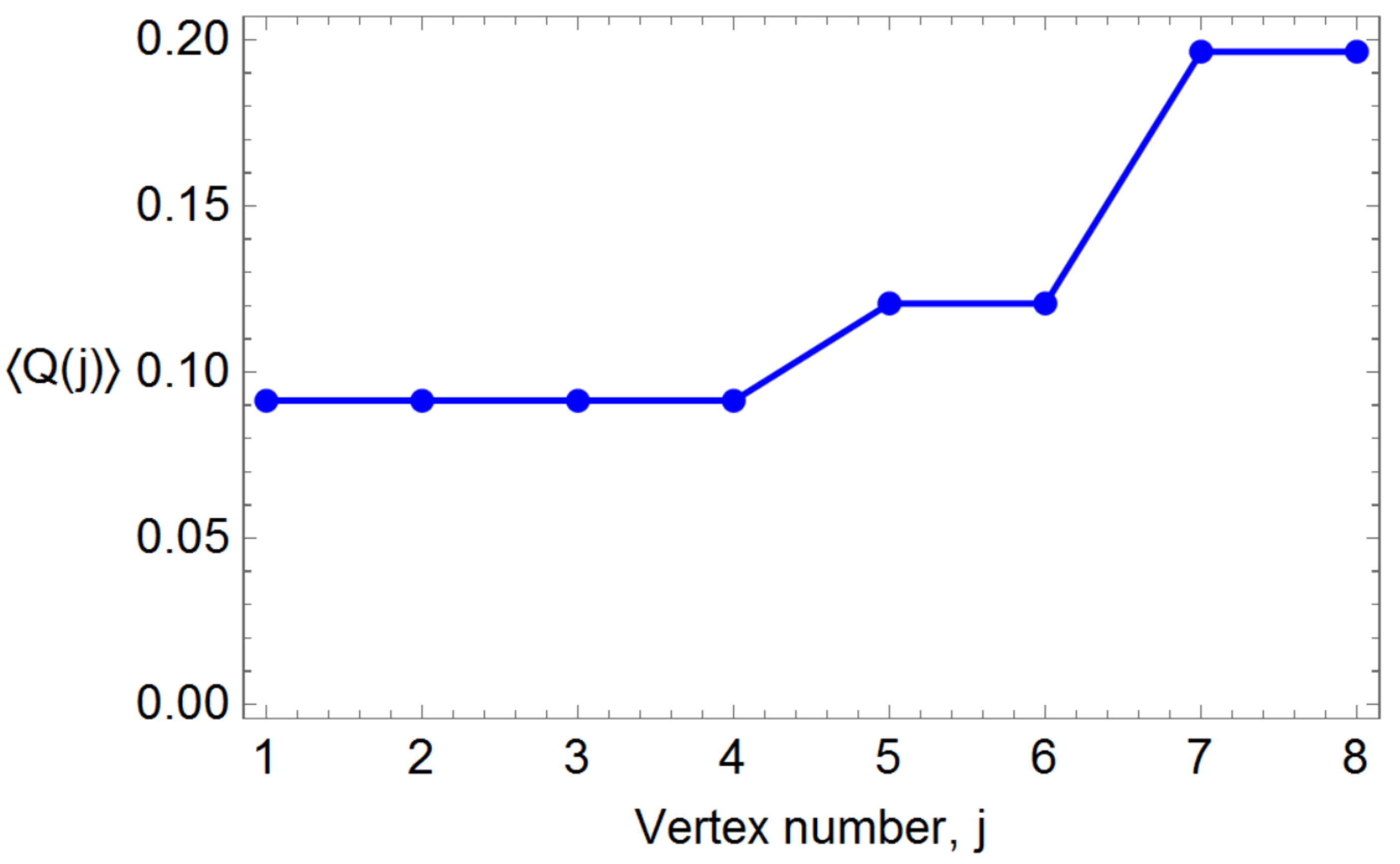}}
	\caption{Instantaneous and average quantum Pagerank results for the directed graph shown in Figure \ref{fig:Gdirected}. In (a), the red, blue and green curves denote the instantaneous quantum Pagerank for vertices 1-4, 5-6 and 7-8 respectively.}
	\label{fig:Gres}
\end{figure}

%For the complete graph (with self-loops) defined by $A_{i,j} = 1$, this gives $\ket{\phi_i} = \ket{\phi_j}$ (i.e. every column is equal), and so we can simply take $T_i = I_N$. Choosing $\ket{b} = \ket{1}$, for $N=2^n$, the routine becomes $K_b = H^{\otimes n}$ (where $H$ is the Hadamard gate) - for other dimensions $ N < 2^n $ where $n = \lceil \mbox{log}_2 N \rceil$, the state preparation method using integrals can be used, since
%\begin{equation}
%P_{j,r} = 
%\begin{cases}
%1 & 1 \leq j \leq N \\
%0 & N < j \leq 2^n
%\end{cases}
%\end{equation}
%
%is a concave function over a convex domain that is efficiently integrable.

\section{Conclusion}
\label{sec:conclusion}

In summary, we have presented a scheme that can be used to construct efficient quantum circuits for Szegedy quantum walks if the transition matrix of the Markov chain possesses translational symmetry in the columns and if the reference state $\ket{\phi_r}$ (or states $\ket{\phi_{r_x}}$) can be prepared efficiently. This scheme, which applies to both sparse and non-sparse matrices, allows for the efficient realization of quantum algorithms based on Szegedy quantum walks. We have identified the class of cyclic permutations and complete bipartite graphs to be amenable to this scheme, as well a class of weighted interdependent networks. We have also applied our formalism to a tensor product of Markov chains, which further extends the classes of Markov chains for which the Szegedy walk can be efficiently simulated. Lastly, we have applied our results to construct efficient quantum circuits simulating Szegedy walks used in the quantum Pagerank algorithm, providing a means to experimentally demonstrate the quantum Pagerank algorithm.

A potential area for further research would be identifying other useful classes of transition matrices to which the formalism of section \ref{subsec:circuit} can be applied. To begin with, we can generalize the idea in section \ref{sec:cyclic}. Suppose we have some preparation routine $K_b \ket{b} = \ket{\phi_0}$. Consider the class of transition matrices where $\ket{\phi_{i+1}} = U \ket{\phi_i}$, where $U$ is some unitary operation. If the operations $\{U,U^2,U^4,U^8,\ldots\}$ can be efficiently implemented, then the Szegedy walk corresponding to the transition matrix can always be efficiently realized. In the case of section \ref{sec:cyclic}, we had $U=R^x$ or $U=L^x$ for some $x\in\mathbb{Z}_+$. We can obtain different classes of transition matrices by changing $U$---choices such as rotation matrices $R_y(\theta)$ and phase shift matrices $R_z(\phi)$ satisfy the required condition, since $R_y(\theta)^x = R_y(x\theta)$ and $R_z(\phi)^x = R_z(x\phi)$ for any $x\in\mathbb{Z}_+$. 

%T.L. is supported by the International Postgraduate Research Scholarship, Australian Postgraduate Award and the Bruce and Betty Green Postgraduate Research Top-Up Scholarship.

\clearpage

\bibliographystyle{model1-num-names}

\bibliography{References}

\clearpage

\end{document}